\documentclass[12pt,a4paper,british]{iopart}

\usepackage{graphicx}
\usepackage{dcolumn}
\usepackage{bm}
\usepackage{braket}
\usepackage{placeins}
\usepackage{dsfont}
\usepackage{multirow}
\usepackage{array}
\usepackage[table]{xcolor}

\usepackage{upgreek}
\usepackage[normalem]{ulem}

\usepackage{iopams}
\expandafter\let\csname equation*\endcsname\relax
\expandafter\let\csname endequation*\endcsname\relax
\usepackage[tbtags]{mathtools}
\usepackage{cases} 

\usepackage{amsthm}
\theoremstyle{plain}
\newtheorem{theorem}{Theorem}
\newtheorem{lem}[theorem]{Lemma}

\newtheorem*{corollary}{Corollary}

\usepackage[numbers,square,sort&compress]{natbib}

\definecolor{myurlcolor}{rgb}{0,0,0.7}
\definecolor{myrefcolor}{rgb}{0.8,0,0}
\usepackage[unicode=true,pdfusetitle, bookmarks=false,bookmarksnumbered=false, bookmarksopen=false, breaklinks=false,pdfborder={0 0 0},backref=false, colorlinks=true, linkcolor=myrefcolor,citecolor=myurlcolor,urlcolor=myurlcolor]{hyperref}

\usepackage{tocloft}

\usepackage{etoolbox}
\makeatletter
\newcommand{\mainmatter}{%
  \setcounter{footnote}{0}%
  \patchcmd{\@makefntext}{\fnsymbol}{\arabic}{}{}%
  \patchcmd{\@thefnmark}{\fnsymbol}{\arabic}{}{}%
  \def\@makefnmark{\textsuperscript{\arabic{footnote}}}%
}
\makeatother

\DeclareGraphicsExtensions{.png,.pdf,.eps,.jpg}
\graphicspath{ {./figs/} }

\makeatletter

\DeclarePairedDelimiterX\expectedadap[2]{\langle}{\rangle}{#1}
\DeclarePairedDelimiterX\braketadap[2]{\langle}{\rangle}{#1 \delimsize\vert #2}

\newcommand*\oline[1]{%
   \vbox{%
     \hrule height 0.5pt
     \kern0.25ex
     \hbox{%
       \kern-0.2em
       \ifmmode#1\else\ensuremath{#1}\fi
       \kern0em
     }}}

\newcommand{\tinyspace}{\mspace{1mu}}

\renewcommand{\vec}{\operatorname{vec}}

\newcommand{\trace}{\mathrm{Tr}}
\newcommand{\I}{\mathds{1}}
\newcommand{\E}{\mathds{E}}
\newcommand{\D}{\mathcal{D}}
\newcommand{\HOp}{\mathcal{H}}

\newcommand{\Trans}{\mathrm{T}}
\newcommand{\Unitary}{\mathcal{U}}
\renewcommand{\S}{\mathcal{S}}

\newcommand{\trm}[1]{\textrm{#1}}
\newcommand{\mrm}[1]{\mathrm{#1}}
\renewcommand{\tr}[1]{\mathrm{Tr}\!\left\{#1\right\}}

\newcommand*\bigcdot{\mathpalette\bigcdot@{.5}}
\newcommand*\bigcdot@[2]{\mathbin{\vcenter{\hbox{\scalebox{#2}{$\m@th#1\bullet$}}}}}

\newcommand{\eqrefs}[2]{(\ref{#1}-\ref{#2})}

\renewcommand{\vec}[1]{\boldsymbol{#1}} 
\newcommand{\vecvar}[1]{\mathbf{#1}} 
\renewcommand{\mat}[1]{\mathbf{#1}} 
\newcommand{\cc}{\textrm{(c)}} 
\newcommand{\Jqvec}[1]{\hat{\vec{#1\phantom{i}}}\hspace{-.2cm}} 

\newcommand{\meanA}[1]{\E\!\left[#1\right]} 
\newcommand{\meanB}[2]{\E_{#1}\!\left[#2\right]} 
\newcommand*{\mean}[1]{%
  \@ifnextchar\bgroup
    {\meanB{#1}}
    {\meanA{#1}}%
}

\newcommand{\var}[1]{\mathrm{Var}\!\left[#1\right]}
\renewcommand{\MSE}[1]{\Delta^2#1}
\newcommand{\est}[1]{\tilde{#1}}

\newcommand{\F}{F}
\newcommand{\gmr}{\upgamma} 
\newcommand{\larmorfreq}{\omega_\text{L}} 
\newcommand{\Tcoh}{T_\text{coh}} 

\makeatother

\begin{document}

\title{Noisy atomic magnetometry in real time}

\author{J\'{u}lia Amor\'{o}s-Binefa and Jan Ko\l{}ody\'{n}ski}
\address{Centre for Quantum Optical Technologies, Centre of New Technologies, University of Warsaw, Banacha 2c, 02-097 Warsaw, Poland}
\ead{j.amoros-binefa@cent.uw.edu.pl \textrm{and} j.kolodynski@cent.uw.edu.pl}
\vspace{10pt}
\begin{indented}
\item[]\today
\end{indented}

\begin{abstract}
Continuously monitored atomic spin-ensembles allow, in principle, for real-time sensing of external magnetic fields beyond classical limits. Within the linear-Gaussian regime, thanks to the phenomenon of measurement-induced spin-squeezing, they attain a quantum-enhanced scaling of sensitivity both as a function of time, $t$, and the number of atoms involved, $N$. In our work, we rigorously study how such conclusions based on Kalman filtering methods change when inevitable imperfections are taken into account:~in the form of collective noise, as well as stochastic fluctuations of the field in time. We prove that even an infinitesimal amount of noise disallows the error to be arbitrarily diminished by simply increasing $N$, and forces it to eventually follow a classical-like behaviour in $t$. However, we also demonstrate that, ``thanks'' to the presence of noise, in most regimes the model based on a homodyne-like continuous measurement actually achieves the ultimate sensitivity allowed by the decoherence, yielding then the optimal quantum-enhancement. We are able to do so by constructing a noise-induced lower bound on the error that stems from a general method of classically simulating a noisy quantum evolution, during which the stochastic parameter to be estimated---here, the magnetic field---is encoded. The method naturally extends to schemes beyond the linear-Gaussian regime, in particular, also to ones involving feedback or active control.
\end{abstract}

%
%
%
\maketitle
%
%

{\pagestyle{plain}
\tableofcontents
\cleardoublepage}

\mainmatter 

\section{Introduction}
\label{sec:intro}
Optical magnetometers based on atomic spin-ensembles~\cite{budker_optical_2007} are considered today as state-of-the-art magnetic-field sensors competing head to head in sensitivity with SQUID-based devices~\cite{clarke2004squid} without need of cryogenic cooling, while being already miniaturised to chip scales~\cite{kitching_chip-scale_2018}. On one hand, they have been demonstrated to be capable of revolutionising medical applications~\cite{jensen_magnetocardiography_2018,boto_moving_2018,limes_portable_2020,zhang_recording_2020}, on the other, when interconnected into global networks, they are used in searches of new exotic physics~\cite{Pospelov2013,pustelny_global_2013}.

Despite constituting a prominent example of quantum sensors~\cite{Degen2017}, the nominal sensitivity of atomic magnetometers is commonly described within their ``classical'' regime of operation as follows~\cite{budker_optical_2007}---also referred to as the ``Equation One''~\cite{budker_sensing_2020}:
\begin{equation}
\Delta^2 B
\;\approx\;
\frac{1}{(\gmr\,\Tcoh)^2} \cdot \frac{\Tcoh}{T} \cdot \frac{1}{N}
\label{eq:Equation_One}
\end{equation}
where the (squared) error in sensing the value of the magnetic field, $B$, simply decreases (quadratically) with respect to the time over which the sensor gathers information about the field while undergoing Larmor precession---with the constant of proportionality given then by the effective gyromagnetic ratio $\gmr$. The maximal value of such probing time, $\Tcoh$, is dictated by the time-scale within which the sensor can preserve its coherence, with $\Tcoh$ being determined by dominant noise-mechanisms exhibited by the atomic ensemble, e.g.~spin-relaxation $\Tcoh=T_1^*$ and/or spin-decoherence $\Tcoh=T_2^*$\footnote{The $*$-notation is typically used to emphasise that these phenomenological quantities refer to the whole ensemble rather than each individual atom~\cite{Wang2005}.}. Moreover, as ``Equation One'' applies (only) in the limit when a sufficient number of measurements are performed over the total sensing time $T$, it contains a division over the number of repetitions, $T/\Tcoh$, in accordance with the central limit theorem. For the same reason, $\Delta^2 B$ is further divided by the number of atoms $N$, as in the ``classical'' regime of operation the atoms within the ensemble can be treated as independent probes.

Although the derivation of ``Equation One'' may be considered ``hand-waving'', it can be formalised by following the so-called \emph{frequentist approach} to estimation theory~\cite{Kay1993}, which, indeed, applies in the limit of \emph{large number of repetitions} (asymptotic statistics), here $T\gg\Tcoh$. In particular, the l.h.s.~of equation \eqref{eq:Equation_One} formally corresponds to the \emph{mean squared error} (MSE) of an (unbiased) estimator of $B$ built basing on all the collected measurement data~\cite{Degen2017}, while its minimum can then be generally associated with the \emph{Cram\'{e}r-Rao Bound} (CRB)~\cite{Kay1993}, which effectively represents the r.h.s.~of equation \eqref{eq:Equation_One}. As a result, one can then answer fundamental questions using techniques of quantum metrology~\cite{Toth2014,Demkowicz2015}, in particular, by how much can the sensitivity \eqref{eq:Equation_One} be improved by allowing for arbitrary quantum states of the atomic ensemble and measurements more general than the natural light-probing scheme based on the Faraday effect~\cite{Budker2002RMP}. This has lead to the seminal observation that ``Equation One'' can be breached---in particular, its $1/N$-behaviour commonly referred to as the \emph{Standard Quantum Limit} (SQL)---by preparing the atomic ensemble in an entangled state~\cite{Pezze2018RMP}, so that the MSE can in principle attain the ultimate \emph{Heisenberg Limit} $\sim1/N^2$~\cite{Giovannetti2004}.

Unfortunately, such claims about the attainable scaling of precision with $N$ have been proven to be overoptimistic in the asymptotic $N$ limit, whenever one accounts for decoherence---noise---whose strength does not effectively decrease with the number of atoms\footnote{For complementary considerations of $N$-dependent decoherence models see e.g.~\cite{Mitchell_2017}.}~\cite{Maccone2011,Demkowicz2015}. In particular, the noise when affecting each atom in an uncorrelated manner constrains the quantum improvement to a constant factor beyond the SQL~\cite{Escher2011,Demkowicz2012}, while when disturbing the whole ensemble in a collective way enforces a positive lower bound on $\Delta^2B$ that cannot be overcome also when letting $N\to\infty$~\cite{Escher2012,Jeske2014,Demkowicz2015}. Nonetheless, for finite but very large atomic ensembles ($N\gtrsim10^5$), multiple optical-magnetometry experiments have spectacularly demonstrated that sensitivities beyond SQL can be reached~\cite{Wasilewski2010,Koschorreck2010,Sewell2012}. In such experiments, the spin-ensemble is prepared in an (entangled) \emph{spin-squeezed state}~\cite{Ma2011} every time before using it to sense the external magnetic field, $B$, which importantly cannot vary over the process of performing the necessary number of experimental repetitions.

The above paradigm, however, must be abandoned if the sensing task considered requires tracking of the magnetic field in \emph{real time}, and one cannot assure the same magnitude of the field to be probed sufficiently many times. This may occur whenever the field follows a predetermined time-varying waveform subject to stochastic fluctuations, or its behaviour in time is simply not known at all~\cite{bar2004estimation}. Still, the \emph{continuous quantum measurement theory}~\cite{jacobs_straightforward_2006} allows one, in principle, to describe then the dynamics of an optical magnetometer conditioned on the data collected in real time~\cite{Thomsen2002,thomsen_continuous_2002}---also beyond the ``classical'' regime---while the ability to perform quantum non-demolition measurements~\cite{Takahashi1999,Kuzmich1999} of the atomic ensemble continuously in time~\cite{Smith2004,Kuzmich2000,Shah2010} appears to be natural. Indeed,
first experiments in this direction have been conducted demonstrating the capability to preserve quantum enhancement when waveforms of a known shape are being probed~\cite{MartinCiurana2017}, or a fluctuating signal is sensed by a magnetometer operating in the ``classical'' regime~\cite{Jimenez2018}. Moreover, the \emph{Bayesian approach} to estimation theory~\cite{Kay1993} provides analogous tools to constrain the attainable ``single-shot'' precision---via the \emph{Bayesian Cram\'{e}r-Rao Bound} (BCRB)~\cite{van2007bayesian} and its variations~\cite{Fritsche2014}---that have been also generalised to the quantum setting~\cite{Tsang2011}.

In case of atomic magnetometry schemes in which spin-squeezing is realised continuously in time by light-probing based on the Faraday effect~\cite{deutsch_quantum_2010}, the Bayesian inference techniques strikingly predict the \emph{average mean squared error} (aMSE)\footnote{Averaged in the Bayesian single-shot scenario over the prior knowledge about the $B$-field dynamics.} to follow at short timescales the Heisenberg limit in absence of decoherence~\cite{Geremia2003,Molmer2004}:
\begin{equation}
\Delta^2 \est{B} \propto \frac{1}{\gmr^2}\cdot \frac{1}{t^3}\cdot\frac{1}{N^2},
\label{eq:ideal_behaviour}
\end{equation}
where in comparison to the (classical) ``Equation One'' \eqref{eq:Equation_One} also the scaling in time is improved from $1/T$ to $1/t^3$---we use above small $t$ instead to emphasise that it now represents the time along a single experimental trajectory. Equation~\eqref{eq:ideal_behaviour} has been subsequently adapted to account for stochastic fluctuations of the field~\cite{Stockton2004,Petersen2006}. However, it has never been rigorously generalised to take into account the impact of decoherence beyond numerical considerations~\cite{Molmer2004} and models that do not incorporate measurement back-action in real time~\cite{Auzinsh2004}, also when allowing for arbitrary quantum measurements to be performed at the end of the sensing procedure~\cite{Albarelli2017,rossi_noisy_2020}.

In our work we tackle this problem by including the collective noise---general anisotropic decoherence affecting the atomic spin-ensemble as a whole---into the analysis, while similarly considering the short time-scales, i.e., the \emph{linear-Gaussian regime}, for which equation \eqref{eq:ideal_behaviour} was derived~\cite{Geremia2003,Molmer2004}. In particular, we focus on the canonical optical magnetometry setup, in which all the atoms are initially polarised along one direction before being continuously spin-squeezed in time via the mechanism of light-probing in a perpendicular direction~\cite{Thomsen2002,thomsen_continuous_2002,Geremia2003}. We explicitly account for the presence of collective noise, as well as stochastic fluctuations of the field, and construct the optimal Bayesian estimator, i.e.~the \emph{Kalman filter} (KF)~\cite{kalman_new_1960,kalman_new_1961}, whose minimal error we are able to resolve in time. Furthermore, we demonstrate how to generalise and adapt the theoretical techniques previously developed to deal with noise within the frequentist approach to quantum metrology~\cite{Demkowicz2012,kolodynski_efficient_2013}, in particular, the \emph{classical simulation} (CS) \emph{method}~\cite{matsumoto_metric_2010}, which allows us to establish fundamental upper bounds on the attainable precision dictated by the decoherence. As a result, we are able to prove that, although at short timescales, when the collective decoherence can be effectively ignored, the error $\Delta^2 \est{B}$ follows the $1/t^3$-behaviour of equation \eqref{eq:ideal_behaviour}, it is subsequently degraded to $1/t$ once the impact of decoherence ``kicks-in''. Moreover, at large times at which the KF attains its optimal performance in tracking the fluctuating field---its \emph{steady state}---the error reaches the minimal value that is determined by both the decoherence rate and the strength of field fluctuations. Focussing instead on the dependence of error on the ensemble size, we show that the Heisenberg-like behaviour $1/N^2$ is similarly lost as $N$ grows and the impact of collective noise becomes significant. In particular, the error always approaches a constant value dictated by the decoherence rate. However, our method allows us to crucially  demonstrate that the precision achieved by the magnetometry setup here considered saturates the ultimate bound for long times and large ensembles, and hence should be considered as the \emph{optimal} realistic scheme to track fluctuating magnetic fields in presence of collective noise.

The manuscript is organised as follows. In section \ref{sec:atomic_sensor_model} we describe the atomic magnetometer of interest, in particular, the corresponding experimental setup as well as the resulting dynamical model of the spin-ensemble being probed continuously in time by light. This allows us to explain in detail in subsection \ref{sec:LG_regime} the linear-Gaussian regime being considered, and discuss the impact of noise on the evolution of the spin-squeezing parameter in subsection \ref{sec:spin_squeez}. The section \ref{sec:Kalman_filter} is then devoted to the explanation and derivation of the Kalman filter as the optimal estimator, as well as the precision it achieves in real time. In section \ref{sec:impact_noise} we present how the classical simulation method can be utilised to derive the ultimate limit on precision dictated by the decoherence, stemming from the Bayesian Cram\'{e}r-Rao Bound. Consequently, in subsection \ref{sec:regimes} we discuss the different regimes the error follows both in time and $N$, and give an intuitive explanation of their origin. Moreover, we demonstrate that the continuous estimation scheme saturates the ultimate limit derived in section \ref{sec:impact_noise} and, hence, may be regarded as optimal. Finally, we conclude our findings in section \ref{sec:conclusions}.

\section{Atomic magnetometer model}
\label{sec:atomic_sensor_model}

\subsection{Setup}
The optical magnetometry scheme we consider, see \fref{fig:atomic_ensemble}a, consists of an ensemble of $N$ spin-1/2 atoms prepared in a \emph{coherent spin state} (CSS) polarized along the $x$-direction~\cite{Ma2011}, 
so that the initial mean and variance of the ensemble angular momentum operators, denoted by the vector $\Jqvec{J}(t)=(\hat{J}_x(t),\hat{J}_y(t),\hat{J}_z(t))^\Trans$ in the Heisenberg picture, read $\braket{\Jqvec{J}(0)}_\trm{CSS} = (J,0,0)^\Trans$ and $\braket{\Delta^2\Jqvec{J}(0)}_\trm{CSS} = (0,J/2,J/2)^\Trans$, respectively, where $J = N/2$.
As shown in \fref{fig:atomic_ensemble}b, one may then naturally visualise the distribution of the angular momentum with help of the Bloch sphere representation. The aim of the scheme is to measure and estimate in real-time the magnetic field $B_t$ being directed along $y$, which fluctuates according to the \emph{Ornstein-Uhlenbeck} (OU) stochastic process~\cite{Gardiner1985}:
\begin{equation}
    dB_t = - \chi B_t\,dt + dW_B,
    \label{eq:OU_process}
\end{equation}
where by $dW_B$ we denote the Wiener noise of zero mean, $\mean{dW_B} = 0$, and variance $\mean{dW_B^2} = q_B \, dt$. The magnetic field $B_t$ induces a Larmor precession around the $y$-axis at the frequency $\larmorfreq(t) = \gmr B_t$, with $\gmr$ being the effective (constant) gyromagnetic ratio. However, as discussed in more detail below, in our analysis we restrict to dynamics at small times in the presence of small magnetic fields, under which the angular momentum operator $\Jqvec{J}(t)$ deviates only slightly from pointing along the $x$-direction---the direction along which the atoms are initially pumped. In principle, this regime could also be achieved over larger time-scales in experiments involving stroboscopic probing~\cite{Vasilakis2011,vasilakis_generation_2015,bao_spin_2020,bao_retrodiction_2020}. However, such implementations would then require the atoms to be initially pumped in the direction perpendicular to the magnetic field, as depicted in \fref{fig:atomic_ensemble}.

\begin{figure}[t]
	\begin{center}
    \includegraphics[width=0.95\textwidth]{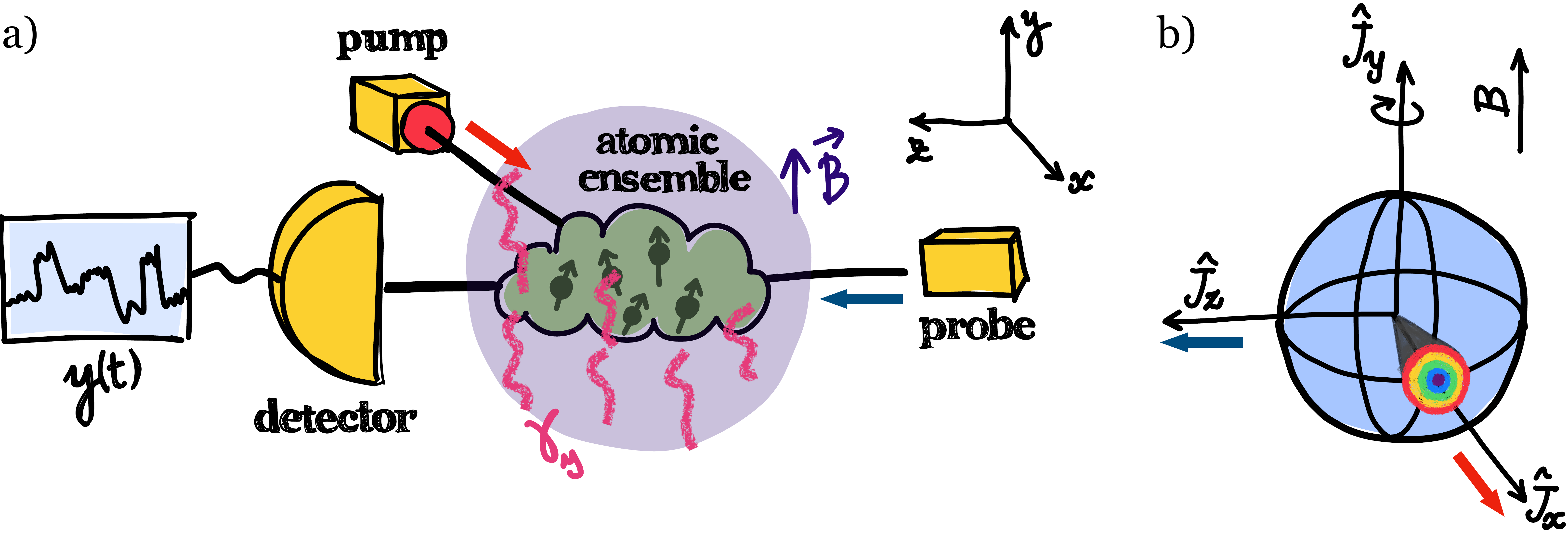}
    \end{center}
    \caption{
    \textbf{Geometry of the atomic magnetometer.} (a) The \emph{magnetometry scheme} involves an atomic ensemble optically pumped along the $x$-direction (\emph{red arrow}) into a coherent spin-state (CSS). The magnetic field being sensed is directed along $y$, while the Faraday--rotation-based continuous measurement is performed by using the light-probe propagating along $z$ (\emph{blue arrow}), which yields a photocurrent signal $y(t)$ being recorded. (b) Bloch sphere representation of the \emph{angular momentum of the ensemble} prepared in a CSS along $x$ (with \emph{red} and \emph{blue arrows} indicating the pumping and probing directions, respectively).
    }
\label{fig:atomic_ensemble}
\end{figure}

\subsection{System and measurement dynamics}
The atomic ensemble is continuously monitored by exploiting the paramagnetic Faraday rotation effect~\cite{Budker2002RMP}, which twists the linear polarisation of the (off-resonant) light propagating through the ensemble in the $z$-direction---the light probe, see \fref{fig:atomic_ensemble}a. As a result, a quantum non-demolition measurement is realised~\cite{Takahashi1999,Kuzmich1999} that is undertaken continuously in time~\cite{Smith2004,Kuzmich2000,Shah2010}, which can be effectively described by means of the homodyne-like continuous measurement~\cite{Wiseman1993} with the shot-noise in the measured signal taking the form of a Wiener process~\cite{deutsch_quantum_2010,handel_modelling_2005}. Note that, even if higher-spin atoms are considered, the following measurement model still applies as long as the contribution from the atomic polarizability tensor component can be suppressed~\cite{deutsch_quantum_2010,Geremia2006,de_echaniz_hamiltonian_2008}. In particular, the photocurrent being measured at time $t$ is proportional to the mean value of the collective atomic spin-component along the probe, $\braket{\hat{J}_z(t)}_\cc$, i.e.:
\begin{align}
    y(t) dt = 2\eta \sqrt{M} \braket{\hat{J}_z(t)}_\cc dt + \sqrt{\eta} \ dW,
    \label{eq:cont_meas}
\end{align}
where $\eta$ is the detection efficiency, $M$ is the measurement strength, and $dW$ is the Wiener differential fulfilling $\mean{dW^2}=dt$ according to the It\^{o} lemma~\cite{Gardiner1985}.

Importantly, due to the active influence the continuous measurement \eqref{eq:cont_meas} exerts on the atoms, the ensemble dynamics becomes conditional---denoted by the subscript $\cc$---and the atoms evolve differently depending on a particular trajectory of the measurement outcomes registered up to (and including) time $t$, $\vec{y}_{\le t}=\{y(\tau)\}_{0\le \tau \le t}$~\cite{handel_modelling_2005}. In particular, the ensemble is characterized at time $t$ by the  quantum state conditioned on the past measurement record, $\rho_\cc(t)\equiv\rho(t|\vec{y}_{\le t})$, which undergoes further \emph{conditional dynamics} described by a stochastic master equation:
\begin{align}
    d\rho_\cc\!(t) = & -i \gmr B_t [\hat{J}_y,\rho_\cc\!(t)]\,dt + \sum_{\alpha = x,y,z} \gamma_\alpha \D{[\hat{J}_\alpha]}\rho_\cc\!(t)\,dt
    \label{eq:cond_dyn} \\
    & + M \D{[\hat{J}_z]} \rho_\cc\!(t)\, dt + \sqrt{M \eta} \HOp{[\hat{J}_z]} \rho_\cc\!(t)\, dW,
    \label{eq:cond_dyn_meas}
\end{align}
where the superoperators $\D$ and $\HOp$ are defined as $\D[\mathcal{O}]\rho = \mathcal{O} \rho \mathcal{O}^\dagger - \frac{1}{2} (\mathcal{O}^\dagger \mathcal{O} \rho + \rho \mathcal{O}^\dagger \mathcal{O} )$ and $\HOp{[\mathcal{O}]}\rho = \mathcal{O} \rho + \rho \mathcal{O}^\dagger - \trace[(\mathcal{O} + \mathcal{O}^\dagger)\rho] \rho$ for any (also non-Hermitian) operator $\mathcal{O}$~\cite{jacobs_straightforward_2006}.

The first term in \eqref{eq:cond_dyn} arises simply from the free Hamiltonian, $\hat{H}=\gmr B_t \hat{J}_y$, responsible for the Larmor precession of the ensemble spin. In contrast, both terms in \eqref{eq:cond_dyn_meas} describe the continuous measurement of $\hat{J}_z$ introduced above~\cite{Thomsen2002,thomsen_continuous_2002}, of which the former is responsible for measurement-induced decoherence, or the backaction, while the latter constitutes the non-linear information gain provided when observing a particular measurement signal in \eqref{eq:cont_meas}~\cite{jacobs_straightforward_2006}. Furthermore, in our analysis we model all possible, e.g.~environment-induced, decoherence mechanisms affecting the ensemble (as a whole) by introducing the second term in \eqref{eq:cond_dyn} with three dissipative components being parametrised by distinct effective rates, $\gamma_\alpha$, in the three directions $\alpha\in\{x,y,z\}$. Let us also emphasise that the Wiener differential, $dW$, appearing in \eqref{eq:cond_dyn_meas} is the same as in the measurement dynamics \eqref{eq:cont_meas}---a particular fluctuation of the photocurrent in \eqref{eq:cont_meas} after being registered drives the conditional state of the ensemble according to \eqref{eq:cond_dyn_meas}.

Finally, let us highlight the difference between the fluctuations of the estimated magnetic field $B_t$, as dictated by the OU process \eref{eq:OU_process}, and the collective decoherence introduced above in \eref{eq:cond_dyn}, of which \emph{only} the latter should be interpreted as a form of ``noise''---in the statistical sense, forcing the quantum state of ensemble to lose its purity over time. In particular, due to its presence, the expectation value of $\hat{J}_\alpha$ ($\alpha = x, y, z$) exponentially decreases with a rate of $\gamma_\alpha$, whose inverse can be identified as the phenomenological (ensemble) spin-decoherence time $T_2^*$~\cite{Wang2005}. Within the theory of open quantum systems~\cite{Breuer2002}, this may be associated with tracing out some environmental degrees of freedom or the environment itself monitoring continuously each component $\hat{J}_\alpha$, whose stochastic trajectory of outcomes is inaccessible and, hence, must be averaged out. On the contrary, the above dynamical model (\ref{eq:OU_process}-\ref{eq:cond_dyn_meas}) describes the evolution along a \emph{single} trajectory of the fluctuating magnetic field. As a result, the impact of field fluctuations on the performance in magnetometry should not be treated on equal grounds with the decoherence of the atomic ensemble, but rather associated with the inability to perfectly estimate the field in real time due to its stochastic nature. In fact, if one on purpose decided to ignore the field fluctuations and rather consider the effective dynamics averaged over all possible field trajectories (see \ref{sec:field_avg}), one would recover the above model of collective decoherence in the direction of the field (here, $\hat{J}_y$), but with a decoherence rate that accumulates as $\gamma_y\sim t^2$ over time.

\subsection{Linear-Gaussian regime}
\label{sec:LG_regime}
In order to define the regime of interest in which the atomic sensor operates, let us first consider the \emph{unconditional dynamics} of the atomic ensemble,
\begin{equation}
    d\rho(t) = -i \gmr B_t [\hat{J}_y,\rho(t)]\, dt + \sum_{\alpha = x,y,z} \gamma_\alpha \D{[\hat{J}_\alpha]}\rho(t)\, dt + M \D{[\hat{J}_z]} \rho(t)\, dt,
    \label{eq:uncond_dyn}
\end{equation}
which can be simply obtained by dropping the stochastic term in \eqrefs{eq:cond_dyn}{eq:cond_dyn_meas}, so that it now describes the evolution of the atomic state at time $t$, averaged over all the stochastic trajectories of the measured outcomes, i.e.~\cite{Gardiner1985}:
\begin{equation}
	\rho(t)=\mean{p(\vec{y}_{\le t})}{\rho_{\cc}(t)}=\int\!d\vec{y}_{\le t}\;p(\vec{y}_{\le t})\;\rho(t|\vec{y}_{\le t}),
\end{equation}
where by $\int\!d\vec{y}_{\le t}$ we denote the integral over all possible measurement records $\vec{y}_{\le t}=\{y(\tau)\}_{0\le \tau \le t}$. Using equation \eqref{eq:uncond_dyn} to define the dynamics of any observable $\hat{O}(t)$ in the Heisenberg picture, whose mean must then obey $\braket{\hat{O}(t)}=\tr{\hat{O}(t) \rho}=\tr{\hat{O} \rho(t)}$, we observe that the (unconditional) evolution of the mean spin-component along the direction of pumping, $\braket{\hat{J}_x(t)}$, is governed by the solution to the following set of coupled differential equations:
\begin{align}
    \label{eq:uncondJx} d\braket{\hat{J}_x(t)} &  = \gmr B_t \braket{\hat{J}_z(t)} dt - \frac{1}{2}(M + \gamma_y + \gamma_z) \braket{\hat{J}_x(t)} dt, \\
    \label{eq:uncondJz} d\braket{\hat{J}_z(t)} &  = -\gmr B_t \braket{\hat{J}_x(t)} dt - \frac{1}{2}(\gamma_x + \gamma_y) \braket{\hat{J}_z(t)} dt, \\
    dB_t & = - \chi B_t \,dt + dW_B \label{eq:OU_process2},
\end{align}
of which the last (stochastic) one just corresponds to the OU process with $\mean{dW_B^2} = q_B$, describing the magnetic-field fluctuations in \eqref{eq:OU_process}.

Although this implies that the dynamics of $\braket{\hat{J}_x(t)}$ is stochastic, we show in \ref{ap:UncondJx} that, if one focuses on short timescales such that $\overline{\larmorfreq(t)}\, t\ll1$, where $\overline{\larmorfreq(t)}=\gmr |\,\oline{B}_t|$ is the time-averaged Larmor frequency, then the mean spin-component along the direction of the pump must decay exponentially as follows:
\begin{align}
    \braket{\hat{J}_x(t)} \approx J e^{- (M + \gamma_y + \gamma_z)\,t/2} = J e^{- r\,t/2}
    \label{eq:Jx_approx}
\end{align}
with an effective \emph{decay rate}:~$r=M + \gamma_y + \gamma_z$. For this to be true, not only the time-average value of the magnetic field, $\oline{B}_t = \frac{1}{t} \int_0^t d\tau B_\tau$, must be small ($|\,\oline{B}_t|\ll\frac{1}{\gmr t}$), but also the distribution of $\oline{B}_t$ should be narrow enough, which we ensure by verifying that $\left|\mean{\oline{B}_t}\pm2\sqrt{\var{\,\oline{B}_t}}\right| \lesssim\frac{1}{\gmr t}$ according to the 68-95-99.7 rule for normal distributions (see also \fref{fig:signal_analysis}a). We explicitly prove in \ref{ap:UncondJx} that for approximation \eqref{eq:Jx_approx} to hold it is sufficient to consider short timescales such that $r\,t\lesssim 1$, while also requiring the parameters of the OU process in \eqref{eq:OU_process2} (or \eqref{eq:OU_process}) to obey:
\begin{equation}
\chi \ll \frac{4\,r}{3}
\qquad\text{and}\qquad
q_B \lesssim \frac{3\,r^3}{4\gmr^2}.	
\label{eq:OU_process_constraints}
\end{equation}

Furthermore, by restricting to short timescales at which the approximation \eqref{eq:Jx_approx} is valid, we effectively deal with spin dynamics in the regime in which $\Jqvec{J}(t)$ only slightly deviates from its initial $x$-polarisation. This allows us then to perform the \emph{Gaussian approximation}~\cite{Molmer2004,Madsen2004} and introduce the effective (time-dependent) quadrature operators:
\begin{align} \label{eq:X&P_quadratures}
    \hat{X}(t) = \frac{\hat{J}_y}{\sqrt{|\braket{\hat{J}_x(t)}|}} \approx \frac{\hat{J}_y}{\sqrt{Je^{-rt/2}}}
    \qquad \text{and} \qquad
    \hat{P}(t) = \frac{\hat{J}_z}{\sqrt{|\braket{\hat{J}_x(t)}|}}\approx \frac{\hat{J}_z}{\sqrt{Je^{-rt/2}}},
\end{align}
which satisfy the canonical commutation relationship $[\hat{X}(t),\hat{P}(t)] = i \frac{\hat{J}_x}{\left|\braket{\hat{J}_x(t)}\right|}\approx i$, as long as $|\braket{\hat{J}_x(t)}| \gg 1$~\cite{Molmer2004,Madsen2004}, which is assured within the approximation \eqref{eq:Jx_approx} whenever $J\gg1$---recall $J=N/2\gtrsim10^5$ in optical-magnetometry experiments~\cite{Wasilewski2010,Koschorreck2010,Sewell2012}. As a consequence, it is enough to consider the evolution of the first and second moments of the quadratures, so that when focussing on the probed spin-component in the $z$-direction that fulfils $\hat{J}_z\approx\sqrt{Je^{-rt/2}}\hat{P}(t)$, its conditional dynamics \eqrefs{eq:cond_dyn}{eq:cond_dyn_meas} is completely described by the following set of equations:
\begin{align}
    y(t) dt & = 2\eta \sqrt{M} \braket{\hat{J}_z(t)}_\cc dt + \sqrt{\eta} \ dW, \label{eq:measure_before} \\
    dB_t & = - \chi B_t \,dt + dW_B \qquad\qquad(\text{with}\;\;\mean{dW_B^2} = q_B)  \label{eq:OU_process_before}\\
    d\braket{\hat{J}_z(t)}_\cc & = - \gmr B_t J e^{- r t / 2} dt + 2 \sqrt{\eta M} \braket{\Delta^2\hat{J}_z(t)}_\cc dW, \label{eq:mean_eq_Jz_before}\\
    d\braket{\Delta^2 \hat{J}_z (t)}_\cc & = -4 M \eta \braket{\Delta^2 \hat{J}_z(t)}^2_\cc dt + \gamma_y J^2 e^{-r t} dt, \label{eq:var_eq_Jz_before}
\end{align}
which assume $[\hat{X}(t),\hat{P}(t)]\approx i$ to hold---satisfied whenever $J \gg 1$, $r t \lesssim 1$ and the conditions \eqref{eq:OU_process_constraints} are valid. Note that equations \eqrefs{eq:measure_before}{eq:mean_eq_Jz_before} are linear with respect to one another, while equation \eqref{eq:var_eq_Jz_before} can be solved independently. Moreover, they involve only Gaussian (in particular, Wiener) stochastic-noise terms, and are derived based on the Gaussian approximation \eqref{eq:X&P_quadratures}. Hence, we refer in short to the evolution described by \eqrefs{eq:measure_before}{eq:var_eq_Jz_before} as the (conditional) dynamics of $\hat{J}_z$ in the \emph{linear-Gaussian regime}.

\subsection{Continuous spin-squeezing}
\label{sec:spin_squeez}
Firstly, let us note that within the linear-Gaussian regime, the equations of motion \eqrefs{eq:measure_before}{eq:mean_eq_Jz_before} are independent of the decoherence rate $\gamma_x$, which is thus redundant. This is a consequence of the approximation \eqref{eq:Jx_approx} (see also \ref{ap:UncondJx}) and can be intuitively explained---deviations of $\Jqvec{J}(t)$ from the $x$-direction are then too small for the collective noise generated by the $\mathcal{D}[\hat{J}_x]$-term in equation \eqref{eq:cond_dyn} to have any effect on the quadratures \eqref{eq:X&P_quadratures}. Furthermore, the decoherence rate $\gamma_z$ is also unnecessary, since by transforming the parameters of the continuous measurement \eqref{eq:cont_meas} as follows:~$M\to M - \gamma_z$, $\eta\to\eta M / (M - \gamma_z)$ and $y\to y\sqrt{M/(M-\gamma_z)}$;~we retrieve the conditional dynamics \eqrefs{eq:measure_before}{eq:var_eq_Jz_before} with $\gamma_z = 0$. Hence, the impact of the collective noise generated by the $\mathcal{D}[\hat{J}_z]$-term in equation \eqref{eq:cond_dyn} instead, can always be interpreted and incorporated into a modified form of the continuous measurement \eqref{eq:cont_meas}.

As a result, without loss of generality, we restrict the effective decay rate of $\braket{\hat{J}_x(t)}$ introduced in \eqref{eq:Jx_approx} to take the form $r=M+\gamma_y$, when considering the (conditional) dynamics of $\hat{J}_z$ in \eqrefs{eq:measure_before}{eq:mean_eq_Jz_before}. We then solve for the dynamics of $\braket{\Delta^2 \hat{J}_z (t)}_\cc$ in \eqref{eq:var_eq_Jz_before}, which constitutes a fully decoupled differential equation. For the complete analytical solution we refer the reader to \ref{ap:VarSol}, where we also show that the time-dependence of the (conditional) $\hat{J}_z$-variance (satisfying $\braket{\Delta^2 \hat{J}_z (0)}_\cc=\braket{\Delta^2 \hat{J}_z (0)}_\trm{CSS}=J/2$ at $t=0$) can be described by two distinct short-time and long-time behaviours:
\begin{subequations}
\begin{numcases}{\braket{\Delta^2 \hat{J}_z (t)}_\cc =}
     V_{<t^*}(t) =  \frac{J}{2} \, \frac{1+2 J t \gamma_y}{1+2 J t M \eta} e^{-(M+\gamma_y)t/2}, & \text{if\quad $0 \le t \ll t^{*}$}
    \label{eq:varJz_t<<t*}\\
     V_{>t^*}(t) =  \frac{J}{2} \, \sqrt{\frac{\gamma_y}{\eta M}} \,e^{-(M + \gamma_y)t/2}, & \text{if\quad $t \gg t^{*}$}
    \label{eq:varJz_t>>t*}
\end{numcases}
\label{eq:varJz_t}
\end{subequations}
where $t^{*} = (2J\sqrt{M\gamma_y \eta})^{-1}$ is the transition time between the two regimes.

Importantly, note that $t \ll t^*$ implies $2 J t \gamma_y \ll \sqrt{\gamma_y/(\eta M)}$. Then, if also $\gamma_y < \eta M$, we may infer $2 J t \gamma_y \ll 1$ and approximate $1 + 2 J t \gamma_y \approx 1$ in \eqref{eq:varJz_t<<t*}. As a result, we then recover the noiseless ($\gamma_y=0$) solution for the variance within the short-time regime, previously found by~\citet{Geremia2003}, i.e.:
\begin{equation}\label{eq:Geremia_VarJz}
  \braket{\Delta^2 \hat{J}_z (t)}_\cc =
   \frac{J}{2} \, \frac{1+2 J t \gamma_y}{1+2 J t M \eta} e^{-(M+\gamma_y)t/2}
   \approx
   \frac{J}{2+4 J t M \eta},
\end{equation}
despite the actual presence of the noise ($\gamma_y>0$). In fact, we prove also in \ref{ap:VarSol} that $\braket{\Delta^2 \hat{J}_z (t)}_\cc$ is a non-decreasing function at $t\approx0$ if $\gamma_y \ge \eta M$. Hence, the noise may be considered negligible at small times \emph{only if} $\gamma_y < \eta M$.

The dynamics of $\braket{\Delta^2 \hat{J}_z (t)}_\cc$ enables us to study the phenomenon of spin-squeezing of the atomic ensemble~\cite{Ma2011} by  directly computing the \emph{(conditional) squeezing parameter} introduced by~\citet{Wineland1994} along the pumping $x$-direction:
\begin{equation}
    \xi^2(t)=\frac{\braket{\Delta^2 \hat{J}_z (t)}_\cc}{\braket{\hat{J}_x(t)}^2}\left[\frac{\braket{\Delta^2 \hat{J}_z(0)}_\trm{CSS}}{\braket{\hat{J}_x(0)}^2_\trm{CSS}}\right]^{-1}=\frac{2J \braket{\Delta^2 \hat{J}_z (t)}_\cc}{\braket{\hat{J}_x(t)}^2}
    \label{eq:xi2}
\end{equation}
which when satisfying $\xi^2(t)<1$ indicates a gain in interferometric precision over the CSS~\cite{Wineland1994}. In particular, just like we decomposed $\braket{\Delta^2 \hat{J}_z (t)}_\cc$ according to \eqref{eq:varJz_t}, we can similarly split the behaviour of the squeezing parameter in the linear-Gaussian regime, in which $\braket{\hat{J}_x(t)} \approx J e^{- r\,t/2}$, as follows
\begin{subequations}
\label{eq:squeezing_cases}
\begin{numcases}{\xi^2(t)
    \approx\frac{2\braket{\Delta^2 \hat{J}_z (t)}_\cc}{Je^{-(M+\gamma_y)t}}
   	=}
    \xi^2_{<t^*}(t) = \frac{1+2 J t \gamma_y}{1+2 J t M \eta}e^{(M+\gamma_y)t/2}, & \text{if\quad $0 \le t \ll t^{*}$,\quad}\\
    \xi^2_{>t^*}(t) = \sqrt{\frac{\gamma_y}{\eta M}}e^{(M+\gamma_y)t/2}, & \text{if\quad $t\gg t^{*}$.}
\end{numcases}
\end{subequations}%
From now on in all plots, in order to confirm that apart from \eqref{eq:OU_process_constraints} the condition $rt \lesssim 1$ (with now $r=M + \gamma_y$) is satisfied and, hence, the linear-Gaussian approximation \eqref{eq:Jx_approx} is valid, we introduce the \emph{rescaled time} $t_S = (M + \gamma_y) t$ and limit its range to $t_S \lesssim 1$, what assures the timescale of the linear-Gaussian regime to always be consistently considered.

\begin{figure}[t!]
  \centering
  \includegraphics[width=\textwidth]{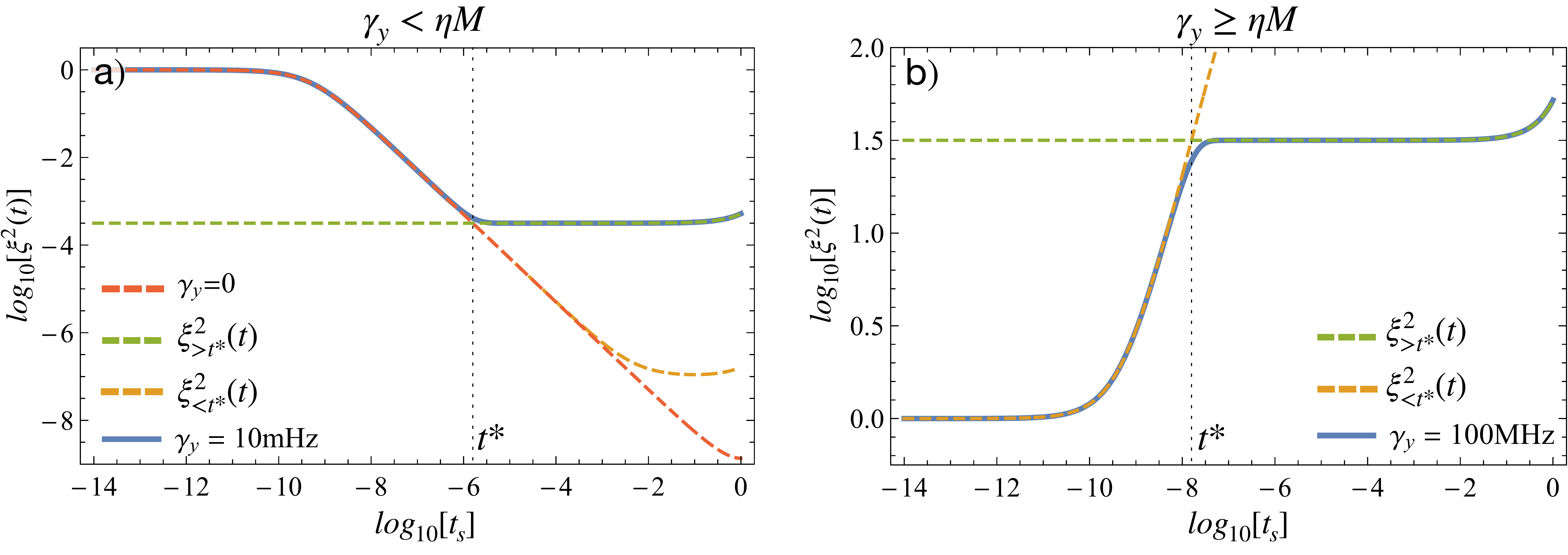}
  \caption{
  \textbf{Evolution of spin-squeezing in time.} The \emph{squeezing parameter} $\xi^2(t)$ defined in equation \eqref{eq:xi2} is plotted as a function of rescaled time $t_S = (M + \gamma_y)t$ for the case of $\gamma_y = 10mHz < M = 100kHz$ (subplot (a)) and $\gamma_y = 100MHz > M = 100kHz$ (subplot (b)). The exact function $\xi^2(t)$ (\emph{solid blue}) is compared with its two different regimes $\xi^2_{<t^*}(t)$ and $\xi^2_{>t^*}$ (\emph{dashed yellow} and \emph{green}, respectively), as well as the noiseless solution when $\gamma_y < M$ (\emph{dashed red}). The other parameters used to generate the plots are:~$J = 10^9$, $\gmr = 1\mrm{kHz/mG}$, and $\eta = 1$.
  }
  \label{fig:squeezing}
\end{figure}

In \fref{fig:squeezing} we present explicitly the exact dynamics of the squeezing parameter \eqref{eq:xi2} for two important cases:~(a) -- when $\gamma_y<\eta M$ and the spin-squeezing ($\xi^2(t)<1$) occurs within a finite-time window (see \fref{fig:squeezing}a);~and (b) -- when $\gamma_y \geq \eta M$ for which spin-squeezing is forbidden (see \fref{fig:squeezing}b). In both cases, it is evident that the exact solution for $\xi^2(t)$ very closely follows the two-regime behaviour in \eqref{eq:squeezing_cases}, with a clear transition at $t\approx t^*$. Moreover, as seen explicitly from the two-regime solution \eqref{eq:squeezing_cases}---and the exact solution that can be straightforwardly derived from the exact evolution of the $\hat{J}_z$-variance \eqref{eq:varJz_t} described in \ref{ap:VarSol}---the dynamics of the squeezing parameter is specified solely by the properties of the continuous measurement ($\eta$ and $M$), collective decoherence ($\gamma_y$), and the number of atoms ($J=N/2$).

\begin{figure}[t!]
    \centering
    \includegraphics[width=.75\textwidth]{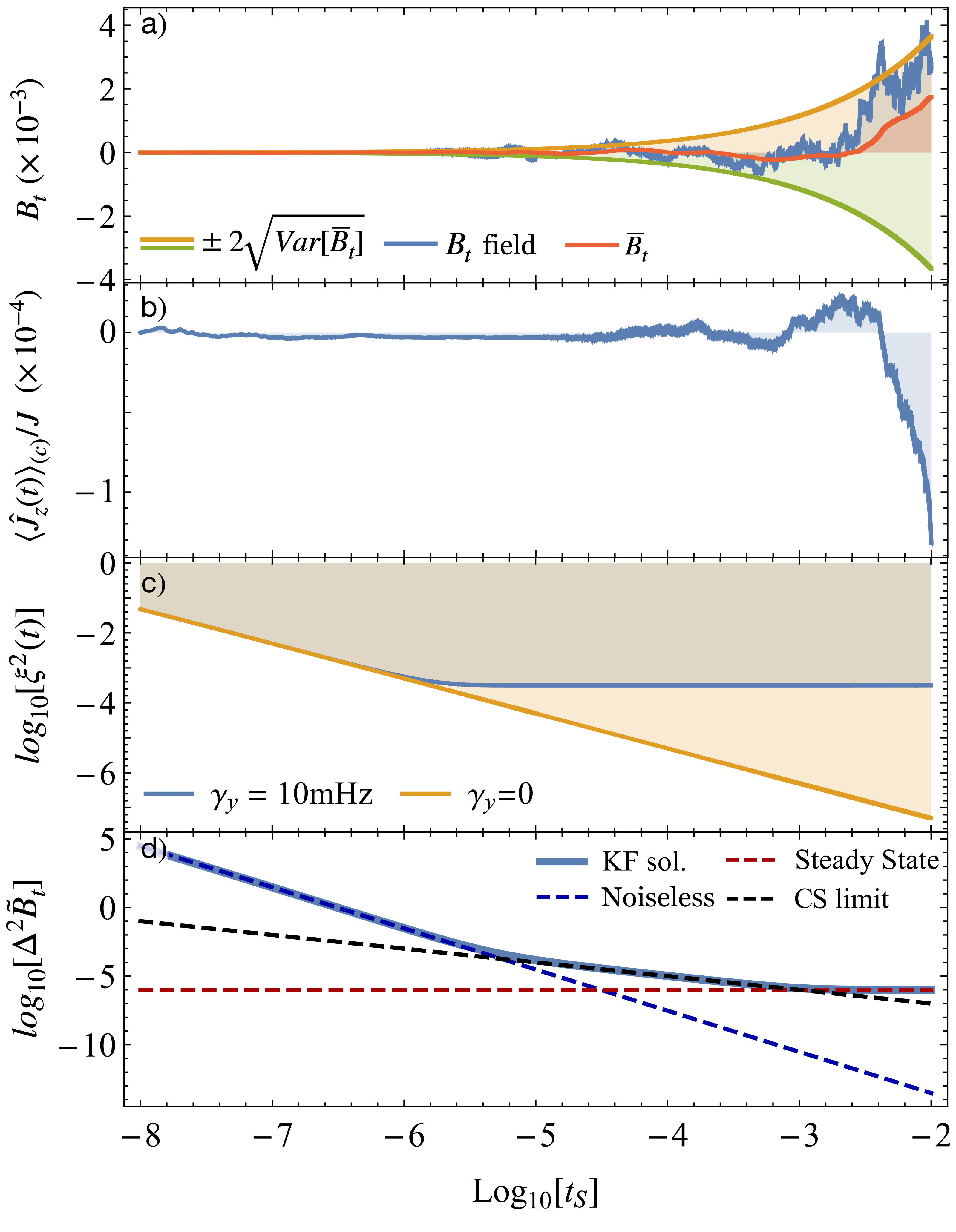}
    \caption{
    \textbf{Dynamics along a single trajectory}.
    Subplot (a) shows the simulated magnetic field along its time average (\emph{solid red}) and confidence intervals (\emph{orange} and \emph{green}) assuring validity of the linear-Gaussian approximation. Subplot (b) presents the normalized conditional evolution of $\hat{J}_z$ driven by the particular B-field trajectory depicted above. Subplot (c) compares the squeezing parameter with (\emph{blue}) and without (\emph{orange}) decoherence present. Finally, subplot (d) juxtaposes the minimal estimation error of the field (\emph{solid blue line}) with its different scalings in time:~noiseless-like ($\propto 1/(t^3 J^2)$, \emph{dashed blue}), ``classical''-like ($\propto 1/t$ predicted by the classical simulation (CS) limit, \emph{dashed black}), and the one dictated by the steady state ($\propto 1$, \emph{dashed red}). All the plots are made with respect to the rescaled time $t_S = (M + \gamma_y) t$, and the other parameters used are:~$\chi = 0$, $q_B = 100G^2/s$, $\eta = 1$,  $\gamma_y = 10mHz$,  $M = 100kHz$, $J = 10^9$, and $\gmr = 1kHz/mG$.}
    \label{fig:signal_analysis}
\end{figure}

In what follows, we turn to the problem of sensing the magnetic field, $B_t$, in real time by constructing the optimal estimator, $\est{B}_t$, of its fluctuating value in the form of a \emph{Kalman filter}~\cite{kalman_new_1960,kalman_new_1961}. However, we have already included in \fref{fig:signal_analysis} the evolution of the corresponding minimal estimation error, $\Delta^2\est{B}_t$, in order to emphasise the fact that the stochasticity of the $B_t$-signal affects the estimation procedure, but \emph{not} the spin-squeezing phenomenon. In particular, for a given magnetic-field trajectory generated according to the OU process \eqref{eq:OU_process}---of a magnitude that does not invalidate the linear-Gaussian regime, see \fref{fig:signal_analysis}a---we obtain a measurement signal $y(t)$ from which we can infer via \eqref{eq:cont_meas} the conditional dynamics of the mean spin-component $\braket{\hat{J}_z(t)}_\cc$, see \fref{fig:signal_analysis}b. Now, the evolution of its variance, $\braket{\Delta^2 \hat{J}_z (t)}_\cc$, allows us to compute the squeezing parameter $\xi^2(t)$, which, as described above, evolves in the linear-Gaussian regime in the same manner for any particular stochastic trajectory of the magnetic field---see \fref{fig:signal_analysis}c. In contrast, the stochastic fluctuations of $B_t$ affect the estimation procedure, so that the Kalman filter reaches a constant error (the \emph{steady state}) after a certain time despite the atomic ensemble still being spin-squeezed ($\xi^2(t)<1$).  This can be directly seen by comparing the evolution of $\Delta^2\est{B}_t$ in \fref{fig:signal_analysis}d against the one of $\xi^2(t)$ in \fref{fig:signal_analysis}c for $t_S\gtrsim10^{-3}$. The transition time when this occurs depends on the parameters of the OU process \eqref{eq:OU_process}---as later shown in \sref{sec:Kalman_steady_state}.

\section{Field estimation with Kalman filtering}
\label{sec:Kalman_filter}
The goal within the magnetometry scheme considered here is to most accurately infer the true value of the magnetic field at a time $t$, given a particular measurement signal $\vec{y}_{\le t}$ recorded up to this time instance. In order to do so, we seek the optimal estimator $\est{B}_t(\vec{y}_{\le t})$ of $B_t$ that minimises the \emph{average (Bayesian) mean squared error} (aMSE)~\cite{Kay1993}:
\begin{equation}
	\label{eq:aMSE_Bt}
    \Delta^2 \est{B}_t = \mean{p(B_t,\vec{y}_{\le t})}{(B_t - \est{B}_t(\vec{y}_{\le t}))^2} = \int \!dB_t\; p(B_t)
    \;\mean{p(\vec{y}_{\le t}|B_t)}{(B_t - \est{B}_t(\vec{y}_{\le t}))^2},
\end{equation}
where by $\int \!dB_t$ we denote the integral over the values that the magnetic field may take at time $t$ (\emph{only}), i.e.~the random variable $B_t$ with probability distribution:
\begin{equation}
	p(B_t)=\int \!dB_0\; p(B_t|B_0)\; p(B_0),
	\label{eq:prior_update}
\end{equation}
which for the OU process \eqref{eq:OU_process} can be determined once the \emph{a priori} distribution $p(B_0)$ is specified, e.g.~a normal distribution of variance $\sigma_0^2$ (see \ref{ap:prior_distribution}). From the Bayesian perspective, $p(B_0)$ should most accurately describe the knowledge an experimentalist possesses about the field at $t=0$, prior to taking any measurements. Note that, as in our work we are interested in estimation of $B$ in real time, we will not consider the setting in which one seeks the optimal estimator for some time $t'<t$ and accounts also for ``future'' measurement-data collected in the time-window $[t',t]$. In that case, one should resort to the inference methods of Bayesian smoothing~\cite{Tsang2010,Zhang2017,Huang2018}, rather than to the filtering ones that are relevant for our purpose.

The optimal estimator minimising the aMSE is always given by the mean of the posterior distribution~\cite{van2007bayesian}, i.e.:
\begin{equation}\label{eq:mMSEest}
  \est{B}_t(\vec{y}_{\le t}) = \mean{p(B_t|\vec{y}_{\le t})}{B_t} = \int \!dB_t\; p(B_t|\vec{y}_{\le t})\,B_t = \mean{p(\vec{B}_{\le t}|\vec{y}_{\le t})}{B_t},
\end{equation}
where in the last step we have used $p(B_t|\vec{y}_{\le t})= \int\!d\vec{B}_{<t}\;p(\vec{B}_{\le t}|\vec{y}_{\le t})$---the fact that the probability of the magnetic field taking at time $t$ the value $B_t$, given a particular measurement trajectory $\vec{y}_{\le t}$, can be obtained by averaging over all past $B$-field stochastic trajectories up to (but not including) time $t$. However, since the problem of interest is described by a set of equations \eqrefs{eq:measure_before}{eq:mean_eq_Jz_before} that are linear and Gaussian, the posterior distribution $p(B_t|\vec{y}_{\le t})$ does not have to be explicitly reconstructed. Instead, the optimal estimator \eqref{eq:mMSEest} can be determined by solving the so-called Kalman-Bucy equation~\cite{van2007bayesian} and is commonly referred to as the \emph{Kalman filter} (KF)~\cite{kalman_new_1960,kalman_new_1961}. The name ``filter'' originates from the fact that the optimal estimator \eqref{eq:mMSEest} at time $t$ can be constructed basing solely on the previous-step estimate and the current measurement outcome $y_t$, so there is, in principle, no need to store all the measurement data.

In order to formulate the problem within the KF-framework~\cite{crassidis2011optimal,sarkka2013bayesian}, we first define the \emph{state} vector of the variables $\braket{\hat{J}_z(t)}_\cc$ and $B_t$ to be estimated\footnote{The KF-framework then naturally provides also the optimal estimate $\braket{\tilde{\hat{J}}_z(t)}_\cc$ of the mean $\braket{\hat{J}_z(t)}_\cc$.} $\vecvar{x}_t = ( \tinyspace \braket{\hat{J}_z(t)}_\cc , B_t \tinyspace )^\Trans$, as well as the \emph{measurement} vector $\vecvar{y}_t=\int_0^t y(\tau) d\tau$, which takes, however, a scalar form with only one d.o.f.~being continuously measured. Consequently, we then identify the corresponding noise vectors as $d\vecvar{w}_t = ( \tinyspace dW,dW_B \tinyspace )^\Trans$ and $d\vecvar{v}_t = \sqrt{\eta} \, dW$, respectively, so that the system of equations \eqrefs{eq:measure_before}{eq:mean_eq_Jz_before} can be compactly rewritten as:
\begin{align}
  d\vecvar{x}_t & = \mat{F}_t \tinyspace \vecvar{x}_t \tinyspace dt + \mat{B}_t  \tinyspace d\vecvar{w}_t,\label{eq:state_system1} \\
  d\vecvar{y}_t & = \mat{H}_t \tinyspace \vecvar{x}_t  \tinyspace dt + d\vecvar{v}_t, \label{eq:state_system2}
\end{align}
where
\begin{equation}\label{eq:matrices_F_K_B}
  \mat{F}_t = \begin{pmatrix} 0 & -\gmr J \tinyspace e^{-r \tinyspace t / 2} \\ 0 & -\chi \end{pmatrix},
  \quad
  \mat{B}_t = \begin{pmatrix}  2 \sqrt{\eta M} \braket{\Delta^2J_Z(t)}_\cc & 0 \\ 0 & 1 \end{pmatrix},
  \quad
  \mat{H}_t =  2\sqrt{\eta M} \begin{pmatrix} 1 & 0 \end{pmatrix}.
\end{equation}
The self- and cross-correlations between stochastic noise-terms are then given by
\begin{align}
	\mean{d\vecvar{w}_t d\vecvar{w}_s^\Trans} & = \mat{Q}_t \tinyspace \delta(t-s) \tinyspace dt, \\
	\mean{d\vecvar{v}_t d\vecvar{v}_s^\Trans} & = \mat{R}_t \tinyspace \delta(t-s) \tinyspace dt, \\
	\mean{d\vecvar{w}_t d\vecvar{v}_s^\Trans} & = \mat{S}_t \tinyspace \delta(t-s) \tinyspace dt,
	\label{eq:crosscorr}
\end{align}
where
\begin{equation}
  \mat{Q}_t = \begin{pmatrix} 1 & 0 \\ 0 & q_B \end{pmatrix},
  \quad
  \mat{R}_t = \eta
  \quad\text{and}\quad
  \mat{S}_t = \begin{pmatrix} \sqrt{\eta} \\ 0 \end{pmatrix},
  \label{eq:noise_correlations}
\end{equation}
are the covariance and cross-correlation matrices (or scalars) of the process and measurement noises, respectively. Let us note that we have used $\mean{dW_B dW} = \mean{dW dW_B} = 0$ in \eqref{eq:crosscorr}, as no cross-correlations are present between the $B$-field fluctuations and the measurement noise; whereas, because $q_B \ge 0$ and $\eta\ge0$, the covariances of process and measurement noises consistently satisfy $\mat{Q}_t,\mat{R}_t\ge0$.

As discussed in more detail in \ref{ap:Corr_KF}, the optimal estimator $\Tilde{\vecvar{x}}_t = (\tinyspace \braket{\tilde{\hat{J}}_z(t)}_\cc, \est{B}_t \tinyspace)^\Trans$ that minimises the overall aMSE---formally corresponding to the trace of the covariance matrix $\mat{\Sigma}_t = \mean{(\vecvar{x}_t - \Tilde{\vecvar{x}}_t)(\vecvar{x}_t - \Tilde{\vecvar{x}}_t)^\Trans}$, i.e.~$\tr{\mat{\Sigma}_t} = \mean{\left\Vert\vecvar{x}_t - \Tilde{\vecvar{x}}_t\right\Vert^2}$---is given by the \emph{correlated} KF~\cite{crassidis2011optimal}. In particular, focussing on the optimal covariance matrix that the filter yields, its dynamics is then determined by a non-linear differential equation of the Riccati form~\cite{crassidis2011optimal}:
\begin{align}
    \frac{d\mat{\Sigma}_t}{dt}\equiv\dot{\mat{\Sigma}}_t \; & = \Big(\mat{F}_t - \mat{B}_t \mat{S}_t \mat{R}^{-1}_t \mat{H}_t \Big) \mat{\Sigma}_t   + \mat{\Sigma}_t \Big(\mat{F}_t - \mat{B}_t \mat{S}_t \mat{R}^{-1}_t \mat{H}_t \Big)^\Trans \nonumber \\
    & \qquad - \mat{\Sigma}_t \mat{H}^\Trans_t \mat{R}^{-1}_t \mat{H}_t \mat{\Sigma}_t   + \mat{B}_t \Big( \mat{Q}_t - \mat{S}_t \ \mat{R}^{-1}_t \mat{S}^\Trans_t \Big) \mat{B}^\Trans_t,
    \label{eq:Riccati_equation}
\end{align}
where the initial condition reads $\mat{\Sigma}_0 = \mathrm{diag}(0,\sigma_0^2)$ with $\sigma_0^2$ being the variance of the Gaussian prior distribution $p(B_0)$ in equation \eqref{eq:prior_update}. In order to allow for analytic solutions in some parameter regimes, we redefine the covariance matrix as $\mat{\Sigma}_t = \mat{Y}_t \mat{X}_t^{-1}$, so that we may rewrite \eqref{eq:Riccati_equation} as a system of coupled but linear differential equations,
\begin{align}
    \label{eq:X2main} \dot{\mat{X}}_t = \ &  - \Big( \mat{F}_t -  \mat{B}_t \mat{S}_t \mat{R}^{-1}_t \mat{H}_t \Big)^\Trans \mat{X}_t  +  \mat{H}^\Trans_t \mat{R}^{-1}_t \mat{H}_t \mat{Y}_t, \\
    \label{eq:Y2main} \dot{\mat{Y}}_t = \ &  \Big( \mat{F}_t - \mat{B}_t \mat{S}_t \mat{R}^{-1}_t \mat{H}_t\Big) \mat{Y}_t + \mat{B}_t\Big( \mat{Q}_t - \mat{S}_t \mat{R}^{-1}_t \mat{S}^\Trans_t \Big) \mat{B}^\Trans_t \mat{X}_t,
\end{align}
with the initial condition now corresponding to $\mat{X}_0 = \I$ and $\mat{Y}_0 = \mat{\Sigma}_0$.

\begin{figure}[t]
    \centering
    \includegraphics[width = 0.65\textwidth]{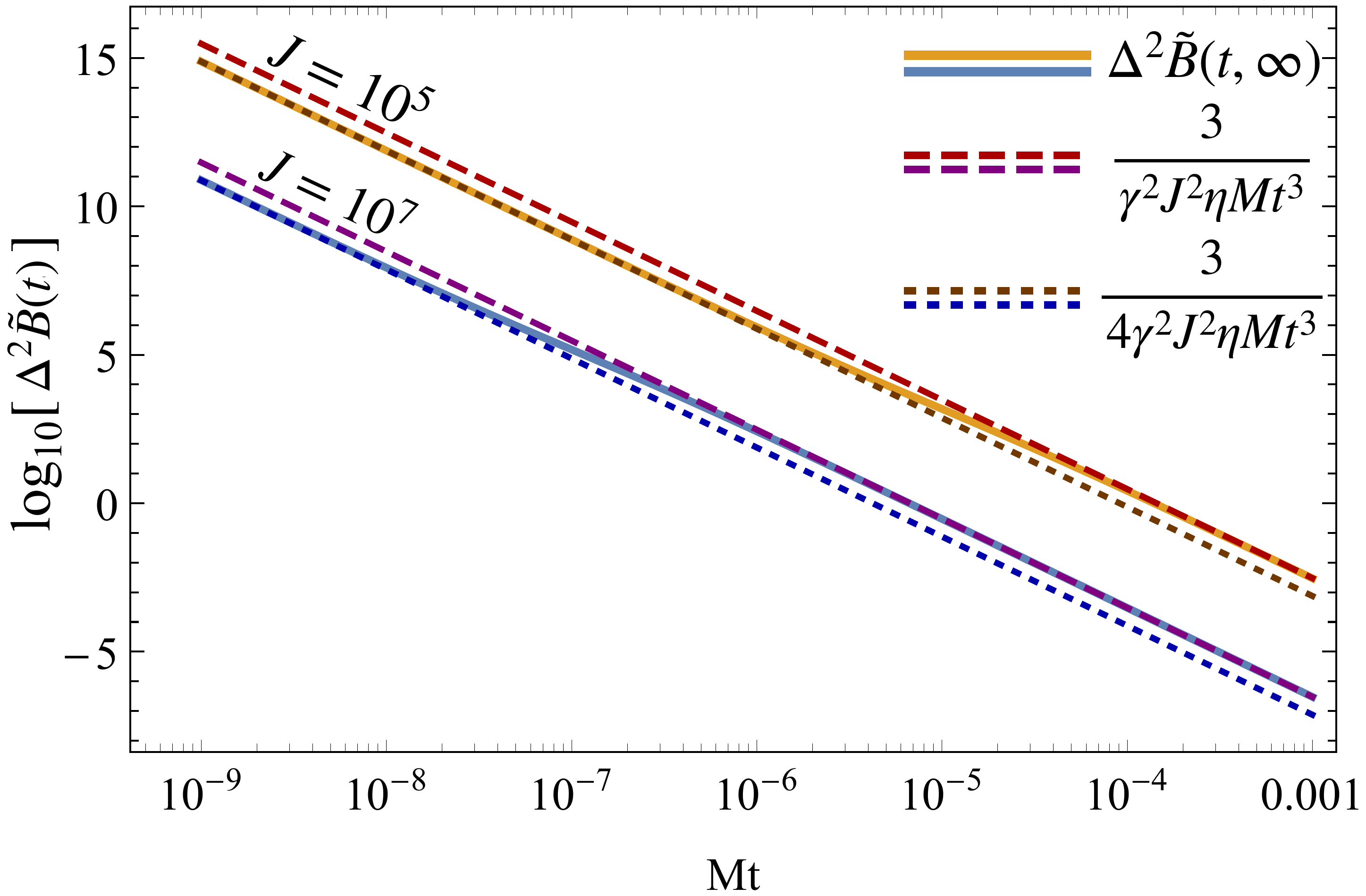}
    \caption{\textbf{Average mean squared error of the Kalman filter (the optimal estimator) in the noiseless scenario}, i.e.~$\Delta^2 \est{B}^\text{HL}_t(\sigma_0)$ of equation \eqref{eq:Geremia_sol_full} for an infinitely wide prior ($\sigma_0\to\infty$) compared with its small- and long-time approximations \eqref{eq:HS_approx}, for parameters:~$M = 100$kHz, $\eta = 1$, and $\gmr = 1$kHz/mG.} 
    \label{fig:Heisenberg_Scaling_fig}
\end{figure}

\subsection{Solution in the absence of decoherence and field fluctuations}
In the absence of collective noise and field fluctuations ($\gamma_y = 0$ and $q_B = \chi = 0$), the Riccati equation \eqref{eq:Riccati_equation} can be explicitly solved~\cite{Geremia2003}. Since $\gamma_y = 0$, the conditional dynamics of the $\hat{J}_z$-variance is then given by equation \eqref{eq:Geremia_VarJz}. Moreover, since no fluctuations of the field are being considered, the noise correlation matrix just reads $\mat{Q}_t = \mathrm{diag}(1,0)$. As a result, the decoupled system of differential equations introduced in \eqrefs{eq:X2main}{eq:Y2main} can be analytically solved, providing the solution for the minimal aMSE in estimating the $B$-field \eqref{eq:aMSE_Bt} as
\begin{align} \label{eq:Geremia_sol_full}
    \Delta^2 \est{B}^\text{HL}_t(\sigma_0) = [\mat{\Sigma}_t]_{22} = \frac{M^2}{16\eta \gmr^2 J^2} \; \frac{(1 + 2JMt\eta)}{a(t) \, e^{-M t} + 4 \, (1+4J\eta) \,  e^{-Mt/2} + b(t)},
\end{align}
where we have emphasised its dependence on the width of the prior, $\sigma_0$, and defined:
\begin{align}
  a(t) & = -(1 + 2\eta J(4 + Mt)), \\
  b(t) & = \frac{M^2}{16 \eta \gmr^2 J^2 \sigma^2_0} + \frac{M^3 t}{8 \gmr^2 J \sigma^2_0} + c(t).
\end{align}
with $c(t)=(Mt-3) + 2 \eta J (Mt-4)$. Note that for an infinitely wide prior, $\sigma_0 \rightarrow \infty$, $b(t) = c(t)$ and the aMSE \eqref{eq:Geremia_sol_full} matches consistently the solution obtained in \cite{Geremia2003} with $\Delta^2 \est{B}^\text{HL}_t(\infty)$ exhibiting the non-classical scaling in both $t$ and $J$---following the Heisenberg limit (HL) $\sim1/N^2$ with $N=J/2$---as mentioned already in equation \eqref{eq:ideal_behaviour}. This becomes clear after making the approximation summarised in \fref{fig:Heisenberg_Scaling_fig}, i.e.:
\begin{subequations}
\label{eq:HS_approx}
	\begin{numcases}{\Delta^2 \est{B}_t^\text{HL}(\infty) \approx}
    	\frac{3}{\eta M t^3} \; \frac{1}{4 \gmr^2 J^2}, & \text{for  $t\ll (JM)^{-1}$,} \label{eq:HS_approxA}\\
        \frac{3}{\eta M t^3} \; \frac{1}{\gmr^2 J^2}, & \text{for $(JM)^{-1}\ll t \quad(< M^{-1})$,} \label{eq:HS_approxB}
	\end{numcases}
\end{subequations}%
where the terms \eqref{eq:HS_approxA} and \eqref{eq:HS_approxB} are obtained by Taylor-expanding $\Delta^2 \est{B}^\text{HL}_t(\infty)$ in equation \eqref{eq:Geremia_sol_full} to leading order in time $t$ and $(J M t)^{-1}$, respectively.

\section{Impact of Noise}
\label{sec:impact_noise}
In \fref{fig:signal_analysis}d we have already presented (solid blue line) the behaviour of the minimal aMSE, $\Delta^2 \est{B}_t$, in presence of noise ($\gamma_y>0$), which we have determined by numerically solving the Riccati equation \eqref{eq:Riccati_equation} after assuming no initial knowledge about the field ($\sigma_0 \to \infty$), and setting $\gamma_y<\eta M$. The latter condition is necessary for the minimal aMSE to initially follow the noiseless solution introduced in \eqref{eq:HS_approx} and, hence, exhibit a ``supraclassical'' scaling in time at short timescales, $\sim 1/t^3$, as only then the spin-squeezing induced by the continuous measurement can occur---the conditional variance of $\hat{J}_z$ decreases with time, see \eqref{eq:Geremia_VarJz} and \ref{ap:VarSol}, so that $\xi^2(t)<1$ in equation \eqref{eq:squeezing_cases}, as depicted in \fref{fig:squeezing}a.

Importantly, although the spin-squeezing is maintained at larger timescales (see \fref{fig:signal_analysis}c), the noiseless behaviour of aMSE in \fref{fig:signal_analysis}d is quickly lost due to noise. In particular, the minimal aMSE firstly follows a classical-like scaling, $\sim1/t$, before attaining a constant value at timescales $t_S \sim 1$, at which still $\xi^2(t)<1$. In the upcoming sections we explicitly show how the presence of collective decoherence and field fluctuations causes the ``supraclassical'' scaling to be lost both in time $t$ and the atom number $N$ ($J=N/2$). We achieve this by giving analytical solutions to what in \fref{fig:signal_analysis}d we refer to as the \emph{classical simulation limit} and the \emph{steady state}.

\subsection{No-go theorem for the Heisenberg limit:~the classical simulation method}
Within the Bayesian approach to estimation theory there exist various general lower bounds on the minimal aMSE, which are commonly referred to as \emph{Bayesian} (or global) \emph{Cram\'{e}r-Rao Bounds} (BCRBs)~\cite{bobrovsky1987}. Still, in the case of linear and Gaussian stochastic processes many of these simplify to a single form~\cite{van2007bayesian,Fritsche2014}. When interested in estimating the process at its last step---here, the magnetic field $B_t$ at the final time $t$--- the applicable bound is the so-called \emph{marginal unconditional} BCRB \cite{Fritsche2014}:
\begin{align} \label{eq:BCRB}
    \MSE{\est{B}_t} \geq (J_B)^{-1},
\end{align}
where $J_B$ is the \emph{Bayesian information} (BI)~\cite{van2007bayesian} evaluated at time $t$:
\begin{align} \label{eq:BI}
    J_B = \mean{p(B_t,\vec{y}_{\le t})}{\left(\partial_{B_t} \log p(B_t,\vec{y}_{\le t}) \right)^2}.
\end{align}
The BI can be conveniently split into two terms:
\begin{align}
    J_B = J_P + J_M,
\end{align}
where $J_P$ represents our knowledge about $B_t$ prior to any estimation, and $J_M$ accounts for the contribution of the measurement record $\vec{y}_{\le t}$. Formally, the BI associated with our prior knowledge, $J_P$, corresponds to the \emph{Fisher information} (FI)~\cite{Kay1993} evaluated with respect to the estimated field value $B_t$ for the distribution $p(B_t)$ defined in \eqref{eq:prior_update}, i.e.:
\begin{align} \label{eq:J_P}
    J_P = \F[p(B_t)] = \mean{p(B_t)}{\left(\partial_{B_t} \log p(B_t) \right)^2}.
\end{align}
In contrast, $J_M$ reads
\begin{align} \label{eq:J_M}
    J_M = \E_{p(B_t,\vec{y}_{\le t})} \left[ \left( \partial_{B_t} \log p(\vec{y}_{\le t},B_t) \right)^2\right] = \int dB_t \; p(B_t) \, \F[p(\vec{y}_{\le t}|B_t)],
\end{align}
where
\begin{equation}
    \F[p(\vec{y}_{\le t}|B_t)]=\mean{p(\vec{y}_{\le t}|B_t)}{\left( \partial_{B_t} \log p(\vec{y}_{\le t}|B_t) \right)^2}
    \label{eq:FI_p(y_t|B_t)}
\end{equation}
is now the FI of the distribution $p(\vec{y}_{\le t}|B_t)$ evaluated again with respect to $B_t$.

The prior contribution to the BI \eqref{eq:J_P} can always be ignored by letting the prior distribution describing our knowledge about the field at $t=0$, $p(B_0)$, be of infinite width~\cite{van2007bayesian}---see \ref{ap:prior_distribution} where we demonstrate that $J_P = 0$ for $\sigma_0^2 \to\infty$. In contrast, the contribution of the measurement record to the BI \eqref{eq:J_M} in general requires one to explicitly evaluate the FI of $p(\vec{y}_{\le t}|B_t)$ for every value $B_t$ may take. In this work, we avoid doing so by constructing an upper-bound on $\F[p(\vec{y}_{\le t}|B_t)]$ that is determined by the noise in the atomic magnetometry scenario, which turns out to be independent of the form of the continuous measurement and the initial state of the atoms.

In order to do so, we discretize the time evolution such that $t = k \, \delta t$---this does not prevent us from obtaining the results in the continuous-time limit, which can be recovered by finally letting $\delta t \to 0$~\cite{jacobs_straightforward_2006}. As a result, the continuous-measurement record $\vec{y}_{\le t}$ can be described by a finite set of (time-ordered) outcomes, $\vec{y}_{k} = \{y_0,y_1, \dots , y_k\}$, collected at the end of each time interval indexed by $j=0,1,\dots,k$. Similarly, the evolution of the magnetic field corresponds now to a discrete-time stochastic process, $\vec{B}_{k} = \{B_0,B_1,\dots,B_k\}$, with $B_j$ representing the value taken by the $B$-field throughout the $j$th time interval. Moreover, the conditional quantum state of the ensemble at time $t=k\delta t$, previously described by $\rho_\cc(t)\equiv\rho(t|\vec{y}_{\le t})$ as the solution to the conditional dynamics \eqref{eq:cond_dyn_meas}, can now be explicitly written as
\begin{align}
    \rho_k\equiv\rho(t|\vec{y}_{k})
    =
    \frac{E_{y_{k}}\,\Lambda_{B_k}\!\left[\dots E_{y_{1}}\,\Lambda_{B_1}\!\left[E_{y_{0}}\,\Lambda_{B_0}\!\left[\rho_{0}\right]E_{y_{0}}^{\dagger}\right]E_{y_{1}}^{\dagger}\dots\right]E_{y_{k}}^{\dagger}}{\tr{E_{y_{k}}\,\Lambda_{B_k}\!\left[\dots E_{y_{1}}\,\Lambda_{B_1}\!\left[ E_{y_{0}}\,\Lambda_{B_0}\!\left[\rho_{0}\right]E_{y_{0}}^{\dagger}\right]E_{y_{1}}^{\dagger}\dots\right]E_{y_{k}}^{\dagger}} },
    \label{eq:con_state_discr}
\end{align}
where the continuous measurement is formally described by a set of measurement operators\footnote{However, the following construction is also valid for schemes in which the form of the continuous measurement (and, hence, operators $E_{y_j}$) changes between time intervals, i.e.~$E_{y_j}^{(j)}\ne E_{y_j}^{(j')}$ for $j\ne j'$.} $E_{y_j}$ that form a positive-operator valued measure (POVM), i.e.~$\{E_{y_j}^\dagger E_{y_j}\}_{y_j}$ with $\sum_{y_j} E_{y_j}^\dagger E_{y_j} =\I$ for every $j$. In the expression \eqref{eq:con_state_discr}, every $\Lambda_{B_j}$ denotes the \emph{quantum channel} (a completely positive trace-preserving map) describing the evolution of the atomic ensemble during the $j$th interval in between measurements. Such evolution will then solely incorporate the Larmor precession under the magnetic field $B_j$ (constant in that time interval), and the collective decoherence. A scheme describing this sequence of channels interspersed with measurements is depicted in \fref{fig:scheme}a (top).

The denominator of the expression \eqref{eq:con_state_discr} should be interpreted as the probability of registering the measurement record $\vec{y}_{k}$ given a particular (discrete) trajectory of the $B$-field $\vec{B}_{k}$, i.e.:
\begin{equation}
    p(\vec{y}_{k}|\vec{B}_{k}) =
    \tr{E_{y_{k}}\,\Lambda_{B_k}\!\left[\dots E_{y_{1}}\,\Lambda_{B_1}\!\left[ E_{y_{0}}\,\Lambda_{B_0}\!\left[\rho_{0}\right]E_{y_{0}}^{\dagger}\right]E_{y_{1}}^{\dagger}\dots\right]E_{y_{k}}^{\dagger}}.
    \label{eq:in_B_out_y}
\end{equation}
Moreover, its \emph{marginal} distribution can then be determined using the Bayes' rule as
\begin{equation}
    p(\vec{y}_{k}|B_{k}) = \frac{p(\vec{y}_{k},B_{k})}{p(B_{k})} = \frac{1}{p(B_{k})}\int d\vec{B}_{k-1}\;p(\vec{B}_{k})\;p(\vec{y}_{k}|\vec{B}_{k}),
    \label{eq:p(y_k|B_k)}
\end{equation}
constituting the discrete-time equivalent of $p(\vec{y}_{\le t}|B_t)$ appearing in equation \eqref{eq:J_M}.

The crucial step in our construction is the observation that the dynamics of the atomic ensemble in between measurements can be \emph{classically simulated} (CS)~\cite{matsumoto_metric_2010,Demkowicz2012,kolodynski_efficient_2013}. In particular, each quantum channel $\Lambda_{B_j}$ in equation \eqref{eq:con_state_discr} can be equivalently replaced by a probabilistic mixture of unitary evolutions such that only the \emph{classical} mixing probability depends on the instantaneous $B$-field value $B_j$---see \fref{fig:scheme}b. In other words, we ``shift'' the encoding of the parameter $B_t$ from a quantum channel to a classical probability, which we show to be Gaussian. Namely, for any $j=0,1,\dots,k$:
\begin{equation}
    \Lambda_{B_j}[\;\bigcdot\;] = \int\! d\omega_j \; q(\omega_j|B_j) \; \Unitary_{\omega_j,\delta t}[\;\bigcdot\;],
\end{equation}
where $\omega_j$ is an auxiliary frequency-like random variable distributed according to the Gaussian distribution:
\begin{equation} \label{eq:in_omega_out_B}
    \omega_j
    \; \sim \;
    q(\omega_j|B_j) = \mathcal{N}\!\left(\gmr B_j,\frac{\gamma_y}{\delta t}\right)\!,
\end{equation}
which parametrises now a \emph{noiseless} Larmor precession of the ensemble within the $j$th time interval described by the unitary transformation:
\begin{equation}
    \Unitary_{\omega_j,\delta t} [\;\bigcdot\;] = e^{-i\omega_j \hat{H} \delta t} \;\bigcdot\; e^{i \omega_j \hat{H} \delta t}
\end{equation}
with the Hamiltonian $\hat{H}=\hat{J}_y$ consistent with the dynamics \eqref{eq:cond_dyn}. A step by step demonstration of this statement is presented in \ref{sec:CS_decomposition}. Moreover, within the setting here considered, the random variable $B_j$ follows a (discrete-time) OU process \eqref{eq:OU_process}, while the initial $B_0$ is drawn from a Gaussian prior distribution of variance $\sigma_0^2$.

As a consequence, by now defining the set of particular frequencies $\omega_j$ chosen in each time interval, $\vec{\omega}_k= \{\omega_0,\omega_1,\dots,\omega_k\}$, we can interpret the conditional distribution of the measurement record \eqref{eq:in_B_out_y} as a probabilistic average over the frequency-like parameters:
\begin{align} \label{eq:equiv}
    p(\vec{y}_{k}|\vec{B}_{k})
    = \mean{q(\vec{\omega}_{k}|\vec{B}_{k})}{p(\vec{y}_{k}|\vec{\omega}_{k})}
    = \int d\vec{\omega}_{k} \; q(\vec{\omega}_{k}|\vec{B}_{k}) \; p(\vec{y}_{k}|\vec{\omega}_{k}),
\end{align}
where $q(\vec{\omega}_{k}|\vec{B}_{k}) = \prod_{j=0}^k q(\omega_j|B_j)$ is just a product distribution and
\begin{align}
    p(\vec{y}_{k}|\vec{\omega}_{k}) =  \trace{\!\left\{ E_{y_{k}}\Unitary_{\omega_{k}\delta t}\left[\dots E_{y_{1}}\Unitary_{\omega_{1}\delta t}\left[E_{y_{0}}\Unitary_{\omega_{0}\delta t}\left[\rho_{0}\right]E_{y_{0}}^{\dagger}\right]E_{y_{1}}^{\dagger}\dots\right]E_{y_{k}}^{\dagger}\right\}}
\end{align}
is now the conditional probability of detecting a set of measurement records $\vec{y}_{k}$, given a particular set of frequencies $\vec{\omega}_{k}$ dictating subsequent unitary evolutions of the atomic ensemble within each time interval.

The decomposition \eqref{eq:equiv} proves the equivalence between a circuit of non-unitary quantum channels describing the noisy dynamics of the atomic ensemble, and the average ($\mean{q(\vec{\omega}_k|\vec{B}_k)}{\dots}$ denoted in \fref{fig:scheme} with $\pmb{\langle}\dots\pmb{\rangle}$) of circuits involving noiseless unitary evolutions, whose frequencies $\vec{\omega}_{k}$ are drawn from a distribution $q(\vec{\omega}_{k}|\vec{B}_{k})$ containing all the information about the field trajectory $\vec{B}_k$. We represent this equivalence schematically in \fref{fig:scheme}a.

\begin{figure}[t]
    \centering
    \includegraphics[width=0.8\textwidth]{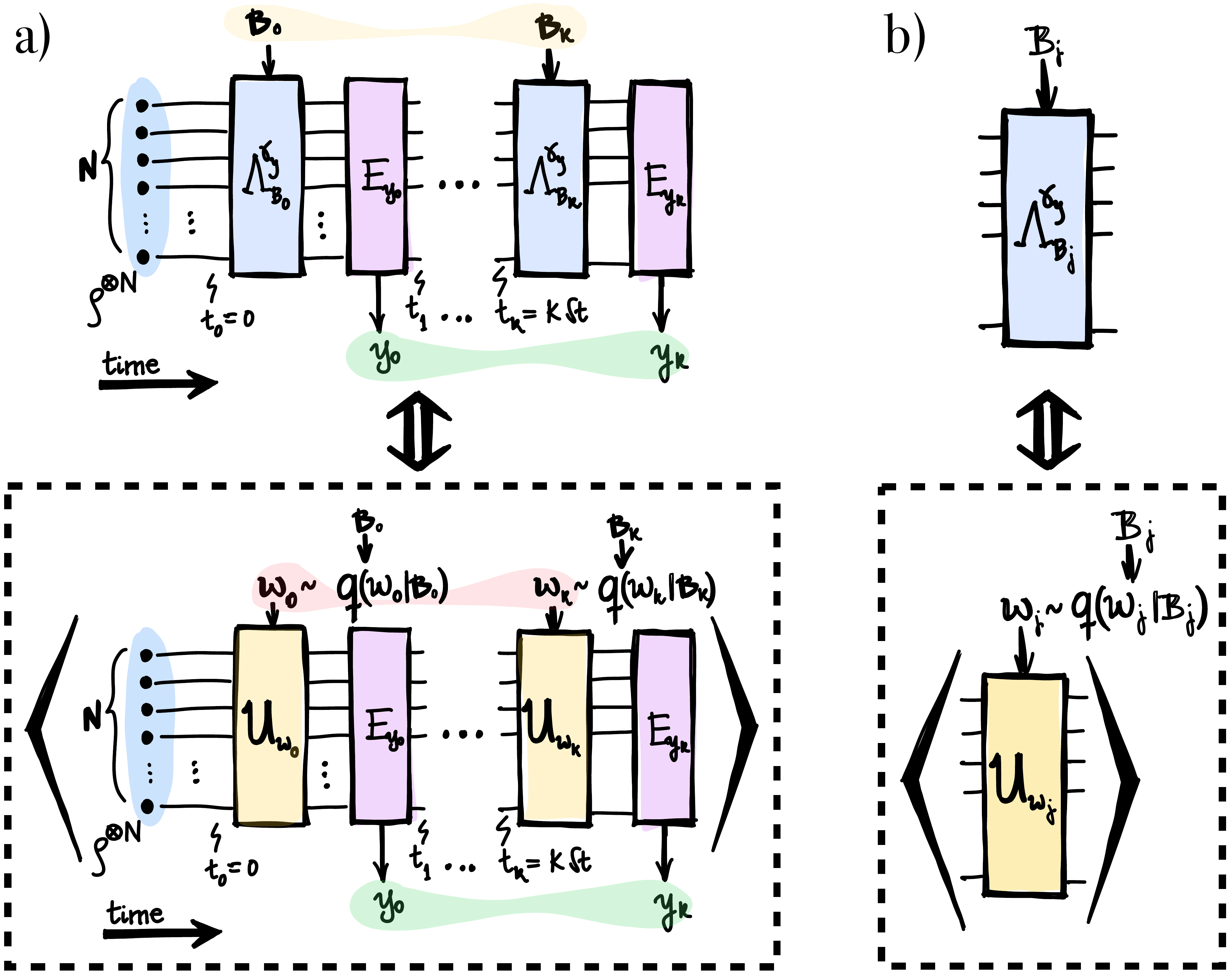}
    \caption{\textbf{Concept of the classical simulation method}. The sensing scheme (a, \emph{top}) involving a continuous measurement can be interpreted as a sequence of quantum channels $\Lambda_{B_i}$ (b, \emph{top})---each dependent on the value of the magnetic field $B_j$---interspersed by instantaneous measurements forming a POVM, $\{E_{y_j}^\dagger E_{y_j}\}_{y_j}$, for every $j=0,1,\dots,k$. As each $\Lambda_{B_j}$ is equivalent to a mixture of unitary channels $U_{\omega_j}$ (b, \emph{bottom}), where the mixing probability $q(\omega_j|B_j)$ has a $B_j$-dependence, the whole scheme can effectively be interpreted as if the consecutive field-values $B_j$ (shaded in \emph{yellow}) are first encoded into the mixing probabilities (a, \emph{bottom}), which only then determine the unitary channels to be applied in the sequence. As a result, any inference strategy with access to all the random variables $\omega_{j}$ (shaded in \emph{red}) can only perform better than the one based on the actual measurement outcomes $y_j$ (shaded in \emph{green}).}
    \label{fig:scheme}
\end{figure}

Focussing on the marginal distribution of interest \eqref{eq:p(y_k|B_k)}, we substitute the decomposition \eqref{eq:equiv} into it, in order to obtain
\begin{align}\label{eq:bigP}
  p(\vec{y}_{k}|B_k) & = \int\! d\vec{\omega}_{k} \;  p(\vec{y}_{k}|\,\vec{\omega}_{k}) \left[\frac{1}{p(B_k)}\int\! d\vec{B}_{k-1} \; p(\vec{B}_{k}) \, q(\vec{\omega}_{k}|\vec{B}_{k})\right]
  = \S_{\vec{\omega}_{k}\to\vec{y}_{k}}\!\left[\mathbb{P}_{B_k}(\vec{\omega}_{k})\right],
\end{align}
where we can now identify $\S_{\vec{\omega}_{k}\to\vec{y}_{k}}[\; \bigcdot \;]=\int\! d\vec{\omega}_{k} \,  p(\vec{y}_{k}|\,\vec{\omega}_{k}) \, \bigcdot $ as a stochastic map independent of the estimated $B_k$, and define the probability distribution function:
\begin{equation}\label{eq:bigP_def}
  \mathbb{P}_{B_k}(\vec{\omega}_{k}) = \frac{1}{p(B_k)} \int d\vec{B}_{k-1} \, p(\vec{B}_{k}) \, q(\vec{\omega}_{k}|\vec{B}_{k}),
\end{equation}
which effectively describes the information about $B_k$ contained within $q(\vec{\omega}_{k}|\vec{B}_{k})$.

Crucially, as the FI is generally contractive (monotonic) under the action of stochastic maps~\cite{Kay1993,Jarzyna2020}, we can now upper-bound $\F[p(\vec{y}_{\le t}|B_t)]$ in \eqref{eq:FI_p(y_t|B_t)} as follows
\begin{align}
    F[p(\vec{y}_{k}|B_k)] & = \F\!\left[\S_{\vec{\omega}_{k}\to\vec{y}_{k}}\!\left[\mathbb{P}_{B_k}(\vec{\omega}_{k})\right]\right] \leq \F \left[  \mathbb{P}_{B_k}(\vec{\omega}_{k}) \right],
\end{align}
and evaluate the FI of $\mathbb{P}_{B_k}(\vec{\omega}_{k})$, which we denote as $\mathbb{P}_{B_t}(\vec{\omega}_{<t})$ upon letting $\delta t \rightarrow 0$ (with $t = k\delta t$) to recover the continuous-time limit. In \ref{sec:bayesian_info_M}, we explicitly show that for the case of infinitely wide ($\sigma_0 \to\infty$) Gaussian prior distribution $p(B_0)$, the FI reads
\begin{equation}
    \F \left[\mathbb{P}_{B_t}(\vec{\omega}_{<t}) \right] = \sqrt{\frac{\gmr^2}{\gamma_y q_B}}\, \tanh\!\left( t \sqrt{\frac{q_B \gmr^2}{\gamma_y}} \right),
    \label{eq:CS_bound_tanh}
\end{equation}
and being independent of the estimated $B_t$ directly sets an upper bound in $J_M$, i.e.:
\begin{align}
  J_M & = \int  dB_t \, p(B_t) \F[p(\vec{y}_{\le t}|B_t)] \\
  & \leq \int dB_t \, p(B_t) \F \left[  \mathbb{P}_{B_t}(\vec{\omega}_{<t}) \right] = \sqrt{\frac{\gmr^2}{\gamma_y q_B}} \, \tanh\!\left( t \sqrt{\frac{q_B \gmr^2}{\gamma_y}} \right).
  \label{eq:finalJM}
\end{align}
Hence, by using also the fact that $J_P = 0$ when $\sigma_0 \to \infty$, we conclude that
\begin{align} \label{eq:LowerBound_BCRB}
    \Delta^2 \est{B}_t \geq \left(J_P + J_M\right)^{-1}
    &= J_M^{-1} \\
    &\geq \sqrt{\frac{\gamma_y \, q_B}{\gmr^2}} \coth{\left( t \sqrt{\frac{q_B \gmr^2}{\gamma_y}}\right)} \equiv \Delta^2 \est{B}_t^\text{CS}(q_B)
    \label{eq:CSlimit_qB} \\
    &\geq \frac{1}{\gmr^2}\frac{\gamma_y}{t} = \Delta^2 \est{B}_t^\text{CS}(0),
    \label{eq:CSlimit_qB_0}
\end{align}
where the last inequality follows from $x\coth( x)\ge1$ for any $x>0$ with the \emph{CS limit}, $\Delta^2 \est{B}_t^\text{CS}(q_B)$, monotonically increasing with $q_B$---as expected, the error is predicted to deteriorate with an increase in the strength of the field fluctuations.

Although we differ the full proof of the upper bound \eqref{eq:finalJM} to \ref{sec:bayesian_info_M}, let us note that in the simplest scenario with the magnetic field being time-invariant, the CS limit can be straightforwardly derived and takes the form $\Delta^2 \est{B}_t^\text{CS}(0)$ as in equation \eqref{eq:CSlimit_qB_0}. For a constant $B$-field, the probability distribution \eqref{eq:bigP_def} simplifies to the \emph{product} $\mathbb{P}_{B}(\vec{\omega}_{k}) = q(\vec{\omega}_{k}|B) = \prod_{j=0}^k q(\omega_{j}|B)$, whose FI is now evaluated with respect to $B$ and reads:
\begin{align}\label{eq:constantB_Fisher}
  \F\left[\mathbb{P}_{B}(\vec{\omega}_{k})\right] & = \F \left[ \prod_{j=0}^k q(\omega_{j}|B) \right] = \sum_{j=0}^{k} \F\left[ q(\omega_{j}|B)\right] =  k\,\F\left[ q(\omega_{j}|B)\right] = \gmr^2 \frac{k \delta t}{\gamma_y} = \gmr^2 \frac{t}{\gamma_y},
\end{align}
so that the CS limit \eqref{eq:CSlimit_qB_0} then directly follows.

In general, the CS limit $\Delta^2 \est{B}_t^\text{CS}(q_B)$ defined in \eqref{eq:CSlimit_qB} can always be approximated independently of the value of $q_B$ by splitting the timescales into two regimes, as follows:
\begin{subequations}
\label{eq:CSlimit}
\begin{numcases}{\Delta^2 \est{B}_t^\text{CS}(q_B) \approx}
     \Delta^2 \est{B}_t^\text{CS}(0)=\frac{1}{\gmr^2}\frac{\gamma_y}{t}, & \text{if\quad$t \lesssim \frac{1}{\gmr}\sqrt{\frac{\gamma_y}{q_B}}$,} \label{eq:CSlimit_lowt} \\
     \Delta^2 \est{B}_{\infty}^\text{CS}(q_B)=\frac{\sqrt{\gamma_y q_B}}{\gmr}, & \text{if\quad$t\gtrsim\frac{1}{\gmr}\sqrt{\frac{\gamma_y}{q_B}} \quad(\lesssim (M+\gamma_y)^{-1})$}, \label{eq:CSlimit_hight}
\end{numcases}
\end{subequations}
where the expression \eqref{eq:CSlimit_lowt} implies that the CS limit \eqref{eq:CSlimit_qB} always simplifies at short times to its form applicable in absence of field fluctuations, i.e.~$\Delta^2 \est{B}_t^\text{CS}(0)$ in \eqref{eq:CSlimit_qB_0}.

Finally, let us highlight that the derivation of the CS limit is independent of the measurement dynamics and the initial state, and depends only on the form of the noisy quantum channel describing the evolution of the system in each discretised step between subsequent measurements.

\subsection{Steady-state solution of the Kalman filter}
\label{sec:Kalman_steady_state}
Returning to the particular measurement model and the estimation strategy considered, we determine the \emph{steady state} (SS) solution of the KF derived in section \ref{sec:Kalman_filter}, as it may then be directly compared with the fundamental CS limit determined above. Still, for simplicity, we focus here only on the SS solution when the magnetic field fluctuates according to the Wiener process~\cite{Gardiner1985}, i.e.~the special case of the OU process \eqref{eq:OU_process} with $\chi = 0$. The complete solution for $\chi > 0$ is presented in \ref{ap:SSS}.

In general, the SS is attained by the KF when its covariance matrix $\mat{\Sigma}_t$ no longer changes in time, so that $\frac{d\mat{\Sigma}_t}{dt}=0$ in equation \ref{eq:Riccati_equation}. This corresponds to the solution of equating the r.h.s.~of \eqref{eq:Riccati_equation} to zero---solving the corresponding (non-linear) algebraic Riccati equation~\cite{crassidis2011optimal}. In our case, finding the SS solution becomes much easier after noting that for $t \gg t^*$ (guaranteed as $t\to\infty$ for SS) $\braket{\Delta^2 \hat{J}_z (t)}_\cc \approx V_{>t^*}(t)$, as in \eqref{eq:varJz_t>>t*}. Then, the SS solution for $\Delta^2\est{B}_t=[\mat{\Sigma}_t]_{22}$ can be shown (see \ref{ap:SSS}) to read
\begin{align} \label{eq:ss_field_sol}
    \Delta^2\est{B}_t^\text{SS} =
     \left( \frac{q_B \, \gamma_y}{\gmr^2} + \frac{1}{\gmr J} \sqrt{\frac{q_B^3}{M \eta}} e^{(M+\gamma_y)t/2} \right)^{1/2},
\end{align}
whose second term in the parenthesis---the one surviving in the absence of noise ($\gamma_y=0$)---has been derived previously in~\cite{Stockton2004}. In contrast, it is the first term arising due to the decoherence considered here ($\gamma_y>0$) that always dominates for large $J$.

In particular, at the relevant timescales $t\lesssim(M+\gamma_y)^{-1}$ whenever $J \gtrsim \frac{\gmr}{\gamma_y}\sqrt{q_B/(\eta M)}$, the second term within the brackets in \eqref{eq:ss_field_sol} becomes negligible and the SS solution \eqref{eq:ss_field_sol} coincides with the CS limit \eqref{eq:CSlimit_hight} valid at long timescales, i.e.:
\begin{equation}\label{eq:ss_to_CS}
 \Delta^2\est{B}_t^\text{SS}
 \;\underset{J\gg1}{=}\;
 \Delta^2 \est{B}_\infty^\text{CS}(q_B)
  = \frac{\sqrt{q_B \gamma_y}}{\gmr},
\end{equation}
which moreover implies that---see also the main plot of \fref{fig:dBvsts}---$\Delta^2\est{B}_t^\text{SS} \approx \Delta^2 \est{B}_t^\text{CS}(q_B)$, as long as further $t \gtrsim \frac{1}{\gmr}\sqrt{\gamma_y/q_B}$ for which the CS limit takes the form \eqref{eq:CSlimit_hight}.

Note that this proves that in presence of decoherence ($\gamma_y>0$) for large enough ensemble size ($J=N/2\,\gg\,1$) the chosen type of continuous measurement (homodyne-like model) and the initial state of the atomic ensemble (CSS) \emph{always} give the best possible precision in estimating the fluctuating magnetic field ($q_B>0$) in the SS regime---bearing in mind that the linear-Gaussian approximation made throughout our work must hold. Crucially, this conclusion holds for large atomic ensembles also at short timescales at which the SS solution does not apply, as we will now explicitly show.

\subsection{Different regimes exhibited by the estimation error}
\label{sec:regimes}
Once the CS limit and the SS solution of the KF have been explicitly derived, we finally have all the information we need to fully understand the evolution of the estimation error in \fref{fig:signal_analysis}d, which we supplement now with additional plots of the aMSE with respect to time and ensemble size ($J=N/2$) in figures \ref{fig:dBvsts} and \ref{fig:dBvsJ}, respectively.

\subsubsection{Estimation error as a function of time}~\\
We depict the time evolution of the estimation error, i.e.~the aMSE of the KF, in \fref{fig:dBvsts}. Importantly, as shown in the main plot, for \emph{``large'' ensembles} (we quantify ``large'' below) the aMSE attains the fundamental CS limit \eqref{eq:CSlimit_qB} that holds for any potential measurement and estimation strategy, hence, proving that the magnetometry scheme achieves then the ultimate precision dictated by the decoherence.

\begin{figure}[t]
    \centering
    \includegraphics[width = .85\textwidth]{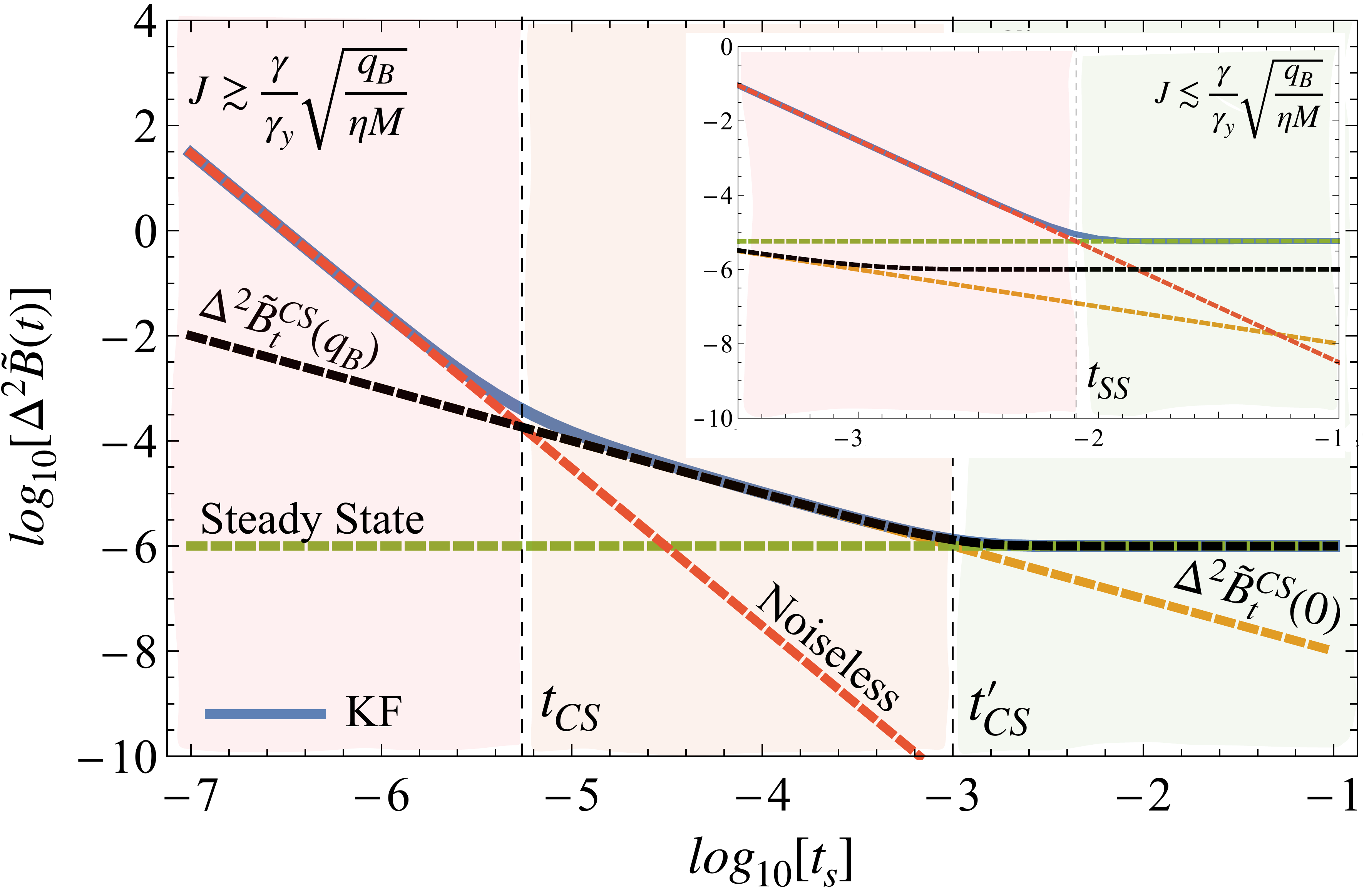}
    \caption{
    \textbf{Estimation error as a function of time for ``large'' (main plot) and ``small'' (inset) atomic ensembles}:~here for $J = 10^9$ and $J = 10^5$, respectively.
    \emph{Solid blue} lines represent the average mean squared error (aMSE) of the Kalman filter (KF), while the dashed lines of different colours denote:~the noiseless solution (\emph{red}), the exact classical simulation (CS) limit (\emph{black}), the CS limit in the absence of field fluctuations (\emph{orange}), and the steady-state (SS) solution of the KF (\emph{green}). Within the main plot the \emph{dashed vertical black} lines mark the transition times, $t_\text{CS}$ and $t_\text{CS}^\prime$  defined in \eqref{eq:transition_ts}, between the noiseless-like ($\propto 1/t^3$, \emph{ros\'{e}}), ``classical''-like ($\propto 1/t$, \emph{orange}) and steady-state ($=\!\sqrt{q_B \gamma_y}/\gmr$, \emph{green}) regimes. Within the latter two, the CS limit \eqref{eq:CSlimit_qB} is importantly saturated. In contrast, for ``small'' ensembles of $J\lesssim \frac{\gmr}{\gamma_y}\sqrt{q_B/(\eta M)}$ (see the inset) the aMSE of the KF saturates, directly after exhibiting the noiseless-like behaviour, the SS solution \eqref{eq:ss_field_sol} at $t_\text{SS}$ defined in \eqref{eq:transition_ts}, without attaining the CS limit at all. Other parameters are set to:~$M = 100$kHz, $\gmr = 1$kHz/mG, $q_B = 100$G$^2$/s, $\gamma_y = 100$mHz, $\chi = 0$ and $\eta = 1$, while the rescaled time is defined as $t_S = (M + \gamma_y)t$.
    }
    \label{fig:dBvsts}
\end{figure}

In such a setting, the behaviour of error as a function of time can be intuitively explained. Firstly, at very short timescales for which the atomic decoherence can be ignored, the aMSE follows the \emph{noiseless-like} (or supra-classical) scaling of $1/t^3$. Once the impact of decoherence ``kicks in'', its behaviour transitions to a \emph{``classical''-like regime} exhibiting $1/t$ scaling, as determined by the CS limit in \eqref{eq:CSlimit_lowt}. At long times, at which the field fluctuations are optimally compensated by the KF, it reaches its \emph{steady state} (SS) and the error saturates at a constant value $\sqrt{q_B \gamma_y}/\gmr$. However, let us emphasise again that this value arises due to the presence of decoherence ($\gamma_y>0$), which in this case dominates the SS solution of the KF in equation \eref{eq:ss_to_CS} and, in fact, assures the error to attain the ultimate CS limit, in particular, its long-time behaviour \eqref{eq:CSlimit_hight}.

In contrast, for \emph{``small'' ensembles}, as presented in the inset, the aMSE does not reach the CS limit. It is so, as the impact of the atomic decoherence can then be actually neglected in comparison to the field fluctuations. As a consequence, the KF reaches quickly with time its SS solution, which is now effectively independent of the decoherence, i.e.~$\Delta^2\est{B}_t^\text{SS}\approx\sqrt[4]{q_B^3/(M \eta \gmr^2 J^2)}$ in equation \eref{eq:ss_field_sol}. One can then interpret the magnetometer to operate in a ``perfect'' manner, with its precision being dictated fully by just the characteristics of the field.

Note that the above two settings are formally distinguished by whether the SS solution in \eqref{eq:ss_field_sol} attains or not the long-time CS limit \eqref{eq:CSlimit_hight}---or, in other words, whether the green line can cross the black line in figure \ref{fig:dBvsts}. Hence, this condition defines naturally what constitutes a ``large'' or ``small'' ensemble above, i.e., $J\gtrsim J'_\text{CS}$ or $J\lesssim J'_\text{CS}$, respectively, where
\begin{equation}
  J_\text{CS}^\prime = \frac{\gmr}{\gamma_y} \sqrt{\frac{q_B}{\eta M}}
  \label{eq:J'_CS}
\end{equation}
follows from the derivation of equation \eqref{eq:ss_to_CS}.

Moreover, by evaluating the times at which the dominant behaviours of the error appearing in figure \ref{fig:dBvsts} cross one another, we explicitly determine the transition times between all the different regimes:
\begin{equation}
  t_\text{CS} = \frac{1}{J}\sqrt{\frac{3}{\eta M \gamma_y}}, \qquad t_\text{CS}^\prime = \frac{1}{\gmr} \sqrt{\frac{\gamma_y}{q_B}}, \qquad  t_\text{SS}  =  \frac{3^{1/3}}{\gmr^2 J^2 \eta M q_B},
  \label{eq:transition_ts}
\end{equation}
which are valid as long as $t\lesssim(M+\gamma_y)^{-1}$ (and the linear-Gaussian approximation holds). In particular, $t_\text{CS}$ marks the transition time of the aMSE from the noiseless-like to the ``classical''-like regime, whose later transition to the steady-state regime occurs then at $t_\text{CS}^\prime$. Similarly, for ``small'' ensembles $t_\text{SS}$ indicates when the aMSE reaches its steady state.

\subsubsection{Estimation error as a function of the number of atoms}~\\
The aMSE as a function of the ensemble size ($J=N/2$) is presented in \fref{fig:dBvsJ}, where we show two distinct scenarios differentiated by the timescales considered. Still, we observe that the ultimate CS limit dictated by the decoherence is attained in both cases as long as a sufficiently large ensemble is considered---proving then the measurement scheme and the estimation procedure to be optimal.

\begin{figure}
    \centering
    \includegraphics[width = .85\textwidth]{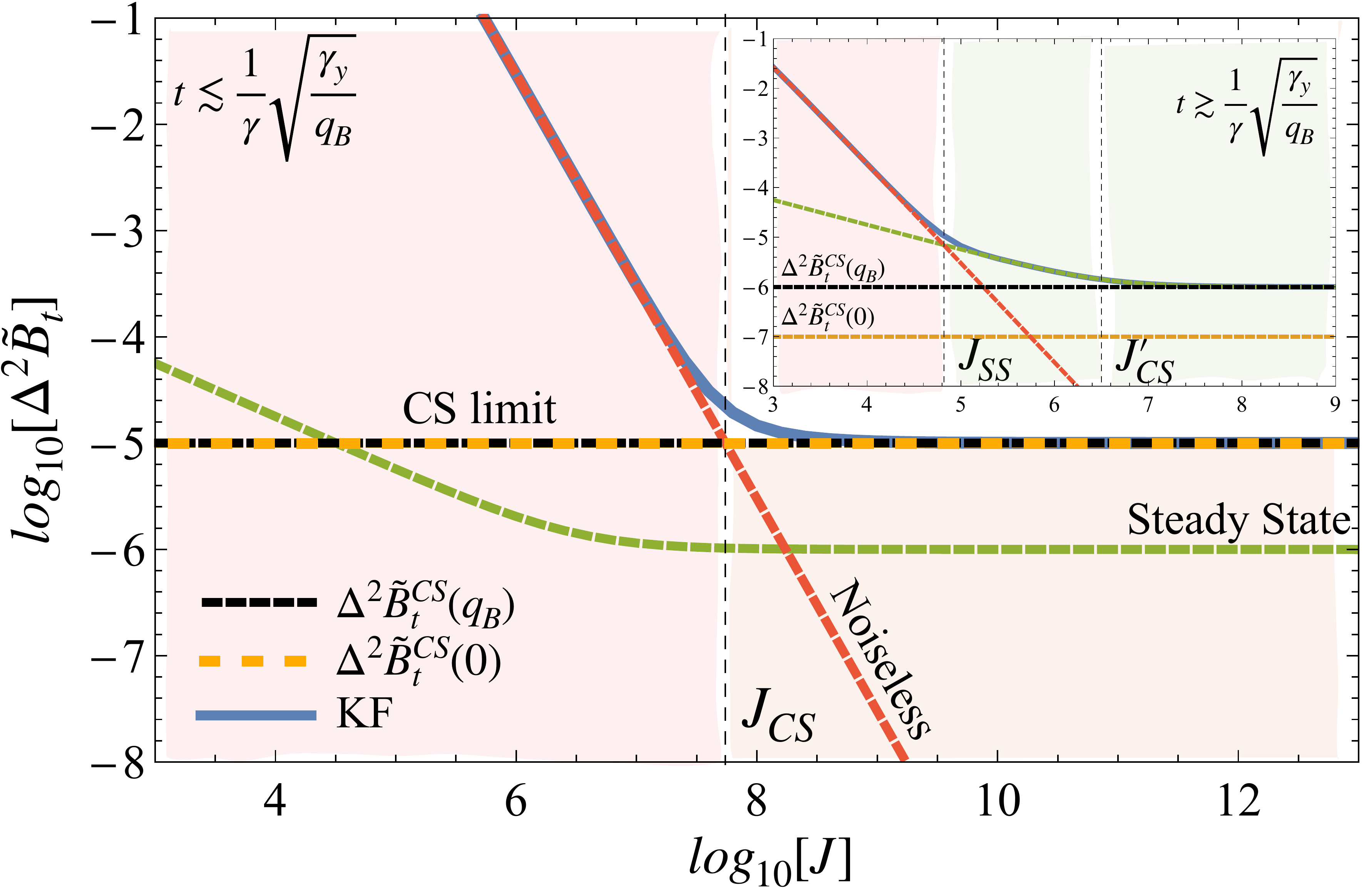}
    \caption{
    \textbf{Estimation error as a function of the ensemble size at ``short'' (main plot) and ``long'' (inset) timescales}: here at $t_S = 10^{-4}$ and $t_S = 10^{-2}$ rescaled times, respectively. The \emph{Solid blue} lines represent the average mean squared error (aMSE) of the Kalman filter (KF), while the dashed lines of different colours denote:~the noiseless solution (\emph{red}), the classical simulation (CS) limit (\emph{black}), the CS limit in the absence of field fluctuations (\emph{orange}), and the steady-state (SS) solution of the KF (\emph{green}). At ``short'' timescales of $t\lesssim \sqrt{\gamma_y/(\gmr^2q_B)}$ (main plot) the SS solution does not apply even for $J\to\infty$. The aMSE, after initially following with $J$ the Heisenberg-like behaviour  ($\propto1/J^2$, shaded in \emph{ros\'{e}}), saturates around $J_\text{CS}$ (defined in \eqref{eq:transition_Js} and marked with a \emph{black dashed vertical} line) at a constant value---given by the CS limit $\gamma_y/(\gmr^2\,t)$ derived in \eqref{eq:CSlimit_lowt}. In contrast, at ``long'' timescales (inset) the aMSE of the KF is explicitly given by its SS solution \eqref{eq:ss_field_sol} for all $J$ above $J_\text{SS}$ defined in \eqref{eq:transition_Js} (\emph{green}-shaded area). As a result, the aMSE, after following with $J$ initially the Heisenberg-like behaviour up to $J_\text{SS}$, firstly scales as $1/\sqrt{J}$ before saturating around $J_\text{CS}^\prime$ defined in \eqref{eq:J'_CS} again at the CS limit, which is now given by the expression \eqref{eq:CSlimit_hight} instead. As in \fref{fig:dBvsts}, we set the other parameters to:~$M = 100$kHz, $\gmr = 1$kHz/mG, $q_B = 100$G$^2$/s, $\gamma_y = 100$mHz, $\chi = 0$ and $\eta = 1$, while the rescaled time is defined as $t_S = (M + \gamma_y)t$.
    }
    \label{fig:dBvsJ}
\end{figure}

In particular, for \emph{``short'' timescales} (main plot), the error as a function of $J$ initially follows the \emph{Heisenberg-like} scaling $1/J^2$ before attaining the short-time CS limit \eqref{eq:CSlimit_lowt} once the effect of the collective noise becomes significant. For \emph{``long'' timescales} (inset), the KF optimally compensates for the field fluctuations and its SS solution applies as long as $J$ is sufficiently large. Still, the long-time CS limit \eqref{eq:CSlimit_hight} is achieved as $J\to\infty$, as it dictates in this case the SS solution \eqref{eq:ss_to_CS} affected by the collective noise ($\gamma_y>0$). The time threshold differentiating between these two cases, $t'_\text{CS}$, is then given by \eqref{eq:CSlimit}---coinciding consistently with the definition in \eqref{eq:transition_ts}.

Importantly, both when the aMSE (solid blue line) attains the CS limit at ``short'' timescales $t\lesssim t'_\text{CS}$ (the main plot of \fref{fig:dBvsJ}), or at ``long'' timescales $t\gtrsim t'_\text{CS}$ (the inset of \fref{fig:dBvsJ}), it saturates at a constant value as $J$ is increased. In each case, that occurs at different ensemble sizes:~$J_\text{CS}$ and $J^\prime_\text{CS}$, respectively. In the latter case, one can also identify a threshold value $J_\text{SS}$, above which the aMSE exhibits a $1/\sqrt{J}$-scaling before saturating at $J^\prime_\text{CS}$. All these can be determined by comparing the dominant behaviours of the error within each regime, and read:
\begin{equation}
  J_\text{CS} = \frac{1}{t} \sqrt{\frac{3}{\eta M \gamma_y}},
  \qquad
  J_\text{SS} = \frac{3^{2/3}}{\gmr t^2 \sqrt{\eta M q_B}},
  \label{eq:transition_Js}
\end{equation}
while $J'_\text{CS}$ coincides consistently with the definition \eqref{eq:J'_CS}.

\section{Conclusions}
\label{sec:conclusions}
We have studied the problem of sensing a magnetic field in real time within the canonical atomic magnetometry setting---a polarised spin-ensemble is being continuously probed in the perpendicular direction to induce spin-squeezing of the atoms, so that a quantum-enhanced precision in estimating the field can be maintained. Within our model we have importantly incorporated both the stochastic fluctuations of the field (in the form of an Ornstein-Uhlenbeck process) as well as collective decoherence (affecting the ensemble as a whole) into the magnetometer conditional dynamics, depending on a particular measurement record collected continuously in time.

As a result, while considering the magnetometer evolution at short timescales within the linear-Gaussian regime, we have computed explicitly the optimal estimator---the Kalman filter---and studied the behaviour of its error both in time, $t$, as well as the effective size of the atomic ensemble, $J=N/2$ with $N$ being the total number of atoms. Moreover, we have developed a \emph{classical simulation method} thanks to which we have established a fundamental limit on the precision that is induced solely by the decoherence. Crucially, as the so-obtained limit applies to any type of state-preparation and continuous-measurement scheme---as long as the conditional dynamics of the magnetometer in between subsequent measurements is not changed---it has allowed us to prove that the continuous measurement model of our interest may often be considered to be optimal in presence of any, even infinitesimal, noise.

In particular, the corresponding average mean squared error (aMSE) of the Kalman filter, which at short timescales always follows the quantum-enhanced behaviour in both time and the number of atoms, i.e.~$1/t^3$ and $1/N^2$, respectively, is bound to saturate eventually the limit dictated by the noise. From the perspective of the time dependence, the noise constrains the aMSE to follow a ``classical''-like $1/t$ scaling before the Kalman filter reaches its steady-state solution---which, in fact, coincides with the ultimate long-time precision induced by the decoherence for large atomic ensembles.
If instead we focus on the dependence of the aMSE with respect to the number of atoms $N$, we observe that the impact of the collective noise is even more drastic, as the aMSE (after following the Heisenberg-like behaviour $1/N^2$) saturates eventually at a constant value---whether or not the Kalman filter operates within the steady-state regime. Furthermore, our analysis allows to straightforwardly estimate both the times and the numbers of atoms at which the transitions between such regimes occur.

Our work paves the way for finding new methods of incorporating effects of decoherence in real-time sensing protocols, while stemming from Bayesian inference techniques combined with tools previously developed within noisy quantum metrology. In particular, as the classical simulation method we have invoked relies on properties of the effective quantum channel describing the dynamics, and not the design of the continuous-measurement scheme under study, e.g.~any adaptive control operations that it incorporates, it should also be directly applicable to sensing protocols involving quantum feedback~\cite{zhang_quantum_2017}. On the other hand, although within our work we have focussed on the estimation task in which only past measurement data may be used for inference (filtering), we believe that our results can be naturally extended to smoothing protocols~\cite{Tsang2010,Zhang2017,Huang2018} that include also retrodiction of data, and have been recently implemented experimentally~\cite{bao_retrodiction_2020,bao_spin_2020}. Moreover, although we have dealt here with the setting of atomic magnetometry, let us stress that the techniques we have presented can also be applied to other real-time sensing platforms, e.g.:~cavity-based experiments with cold atoms incorporating feedback~\cite{hosten_measurement_2016,Cox2016}, requiring similar continuous-measurement theory~\cite{Shankar2019}; or optomechanical devices ~\cite{wieczorek_optimal_2015,rossi_observing_2019,Iwasawa2013} and levitated nanoparticles~\cite{Setter2018,magrini_optimal_2020} that naturally evolve respecting Gaussian dynamics, and hence directly require Kalman-filtering and alike techniques.

\ack{}{}
\addcontentsline{toc}{section}{Acknowledgements}
This research was supported by the Foundation for Polish Science within the “Quantum Optical Technologies” project carried out within the International Research Agendas programme cofinanced by the European Union under the European Regional Development Fund, as well as the QuantERA ERA-NET Cofund in Quantum Technologies implemented within the EU's Horizon 2020 Programme (C’MON-QSENS! project) via the National Science Centre Poland.

\appendix
%
\addtocontents{toc}{\setlength{\cftsecnumwidth}{16ex}}
\addtocontents{toc}{\setlength{\cftsubsecnumwidth}{16ex}}

\section{Unconditional dynamics of $\braket{\hat{J}_x(t)}$ with field fluctuations}
\label{ap:UncondJx}
The unconditional evolution of $\braket{\hat{J}_x(t)}$ can be computed from the stochastic set of differential equations \eqrefs{eq:uncondJx}{eq:OU_process2}, where $B_t$ follows the OU process described in \eqref{eq:OU_process}. Although the full system of differential equations can in principle be solved numerically, we would like to find approximate analytical solutions valid in particular parameter regimes.

In particular, we focus on timescales short enough, such that we can assure $\larmorfreq(t)\, t\ll1$, where $\larmorfreq(t)=\gmr B_t$ is the instantaneous Larmor frequency following the fluctuations of $B_t$. In order to identify and ensure such timescales, we enforce $\overline{\larmorfreq(t)}\, t\ll1$, where $\overline{\larmorfreq(t)}=\gmr \,\oline{B}_t$ is now the time-average over the duration $t$ with
\begin{equation}
    \oline{B}_t = \frac{1}{t} \int_0^t d\tau B_\tau.
    \label{eq:AvB_t}
\end{equation}
Replacing $B_t$ by $\,\oline{B}_t$ within the system of differential equations \eqrefs{eq:uncondJx}{eq:OU_process2} and treating it as a constant, we solve them for $\braket{\hat{J}_x(t)}$ to get
\begin{align}
    \braket{\hat{J}_x(t)} = \frac{J}{2\Theta} \, e^{- (M + \gamma_x + 2\gamma_y + \gamma_z + \Theta)\,t/4} \left(M - \gamma_x + \gamma_z + \Theta - e^{t\Theta/2} (M - \gamma_x + \gamma_z - \Theta ) \right),
\end{align}
where $\Theta = \sqrt{(M-\gamma_x+\gamma_z)^2 - \left(4 \gmr \, \oline{B}_t\right)^2}$. Then, by expanding the above expression to leading order in $\oline{B}_t$, we obtain
\begin{align}
    \braket{\hat{J}_x(t)}
    & \approx
    J e^{- (M + \gamma_y + \gamma_z)\,t/2}\left(1 + 2 \, \gmr^2 \, \oline{B}_t^2 \, \frac{2-2e^{(M-\gamma_x+\gamma_z)t/2}+t(M-\gamma_x+\gamma_z)}{(M-\gamma_x+\gamma_z)^2}\right), \label{eq:ap_UncondJx_Taylor} \\
    & \approx
    J e^{- (M + \gamma_y + \gamma_z)\,t/2}\left(1 +  \frac{\gmr^2 \, \oline{B}_t^2 \, t^2}{2}\right),
    \label{eq:ap_UncondJx_appr1}
\end{align}
where in \eqref{eq:ap_UncondJx_appr1} we have further assumed $t\lesssim(M-\gamma_x+\gamma_z)^{-1}$, so that $e^{(M-\gamma_x+\gamma_z)t/2} \approx 1 + \frac{1}{2}(M-\gamma_x+\gamma_z)t + \frac{1}{8}(M-\gamma_x+\gamma_z)^2t^2$ holds up to the second order in $t$.

Hence, it seems that we may approximate the dynamics of the mean value of the spin-component $\hat{J}_x$ as
\begin{align}
    \braket{\hat{J}_x(t)} \approx J e^{- (M + \gamma_y + \gamma_z)\,t/2},
    \label{eq:ap_UncondJx_appr}
\end{align}
which constitutes the basis for the linear-Gaussian approximation introduced in equation \eqref{eq:Jx_approx}, as long as: (i) $t\lesssim(M-\gamma_x+\gamma_z)^{-1}$ and (ii) $t \lesssim \frac{\sqrt{2}}{\gmr \, |\,\oline{B}_t|}$; of which the latter condition allows us to neglect the term in the parenthesis in \eqref{eq:ap_UncondJx_appr1}. Moreover, as the noise-parameter $\gamma_x$ does not enter the dynamics any more, after letting $\gamma_x\approx0$ and reparametrising the measurement strength as in section \ref{sec:spin_squeez}, $M\to M-\gamma_z$, the condition (i) simplifies to $t\lesssim1/M$ and is naturally satisfied by the more stringent requirement on the rescaled time $t_S=(M+\gamma_y)t$ to always obey $t_S\lesssim1$, which we ensure throughout the main text.

However, we must be more careful when ensuring the validity of condition (ii). As $B_t$ follows a stochastic process, the time-average $\,\oline{B}_t$ defined in \eqref{eq:AvB_t} is a random variable in itself. Therefore, we must further assure that the condition (ii) holds for almost all stochastic trajectories of $B_t$. We achieve this by applying the 68-95-99.7 rule, which allows us to state that if
\begin{equation} \label{eq:t_UB_AvB_t}
	t \lesssim \frac{\sqrt{2}}{\gmr \, \left|\mean{\oline{B}_t} \pm 2\sqrt{\text{Var}[\, \oline{B}_t]}\right|},
\end{equation}
then the approximation \eqref{eq:ap_UncondJx_appr} of \eqref{eq:ap_UncondJx_appr1} is valid with 95\% probability.

Evaluating the mean of $\,\oline{B}_t$ defined in \eqref{eq:AvB_t}, we obtain
\begin{equation}
    \mean{\,\oline{B}_t} = \frac{1}{t} \int_0^t \mean{B_\tau} d\tau
    = \frac{1}{t} \int_0^t B_0 \, e^{-\chi t} d\tau
    = \frac{B_0}{\chi t}\left( 1-e^{-\chi t}\right),
\end{equation}
where $\mean{B_\tau}=B_0 e^{-\chi t}$ for the OU process \eqref{eq:OU_process}~\cite{Gardiner1985}. Similarly,
we calculate the variance of the time-averaged value of the magnetic field, $\,\oline{B}_t$, as
\begin{align}
    \text{Var}[\, \oline{B}_t]
    & = \mean{(\,\oline{B}_t)^2}-\mean{\,\oline{B}_t}^2 \\
    & = \frac{q_B}{2\chi^3 t^2} (4e^{-\chi t} + 2\chi t -e^{-2\chi t} - 3) - \frac{B_0^2}{\chi^2 t^2}\left( 1-e^{-\chi t}\right)^2,	
\end{align}
using the expression for the two-time correlation function of the OU process \eqref{eq:OU_process}, i.e.~\cite{Gardiner1985}:
\begin{align}
    \mean{B_s B_t} = \frac{q_B}{2\chi} \left(e^{-\chi|t-s|} - e^{-\chi(t+s)} \right),
\end{align}
and evaluating
\begin{align}
    \mean{(\,\oline{B}_t)^2}
    &= \frac{1}{t^2} \int_0^t \! d\tau_1 \int_0^t \! d\tau_2 \; \mean{B_{\tau_1}B_{\tau_2}}\\
    &= \frac{1}{t^2} \int_0^t \! d\tau_1 \int_0^t \! d\tau_2 \; \frac{q_B}{2\chi} \left(e^{-\chi |\tau_1 - \tau_2|} - e^{-\chi \tau_1}e^{-\chi \tau_2}\right)\\
    &= \frac{q_B}{2\chi t^2} \left[ \int_0^t d\tau_1 \left( \int_0^{\tau_1} d\tau_2 \, e^{-\chi(\tau_1 - \tau_2)} + \int_{\tau_1}^t d\tau_2 \,  e^{\chi(\tau_1 - \tau_2)}\right) - \left( \frac{1-e^{-\chi t}}{\chi}\right)^2 \right] = \\
    &= \frac{q_B}{2\chi t^2} \left[ \frac{1}{\chi} \int_0^t d\tau_1 \, \left(1 - e^{-\chi \tau_1} - e^{\chi (\tau_1-t)} + 1\right) - \left( \frac{1-e^{-\chi t}}{\chi}\right)^2 \right] = \\
    &=\frac{q_B}{2\chi^3 t^2} (4e^{-\chi t} + 2\chi t -e^{-2\chi t} - 3).
\end{align}

\begin{figure}[t]
  \centering
  \includegraphics[width=0.95\linewidth]{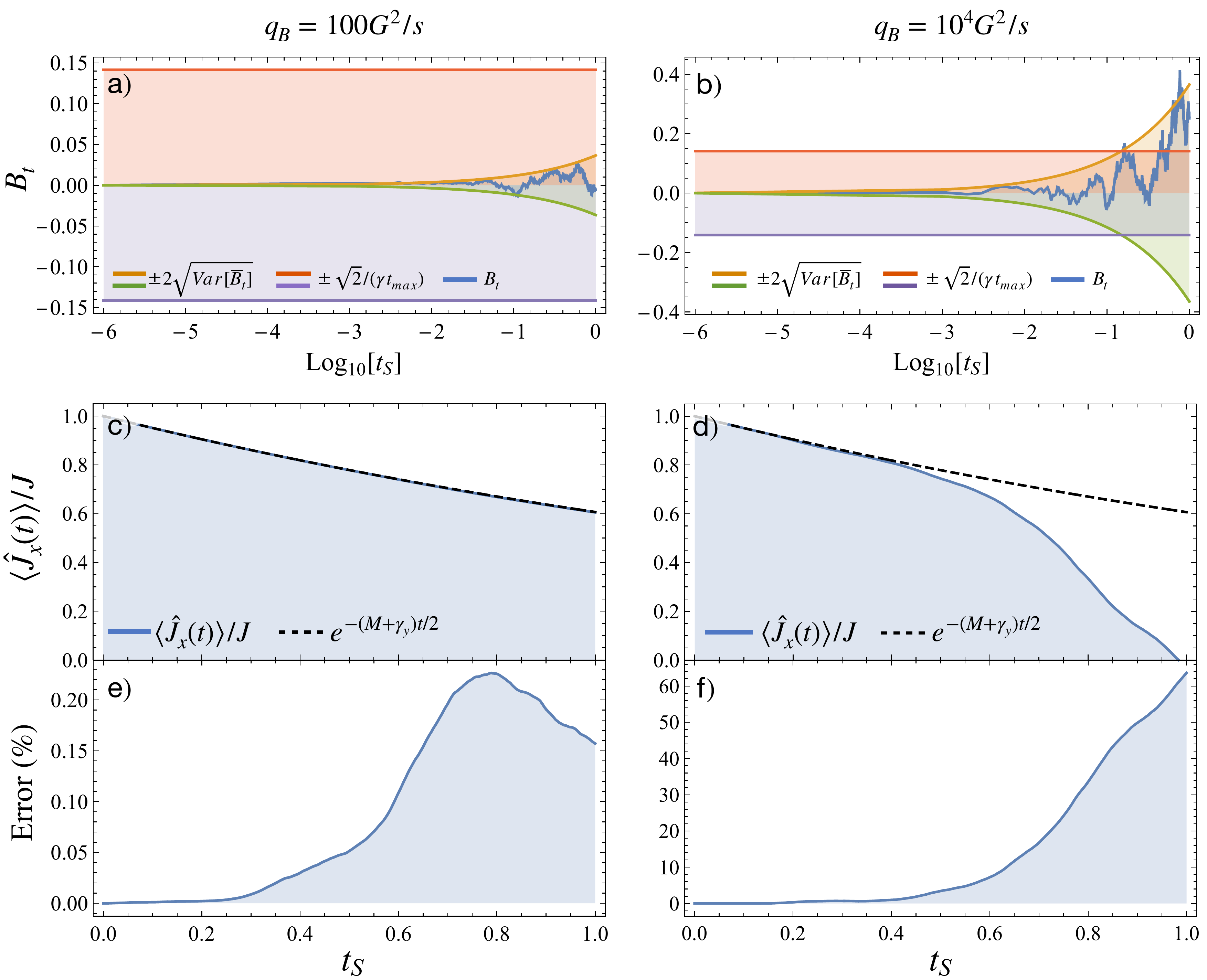}
  \caption{The parameters used to generate the plots are $\gmr = 1kHz/mG$, $M = 100kHz$, $\gamma_y = 1Hz$, $J = 10^7$, $\eta = 1$, and $\chi = 0$, with $t_S = (M + \gamma_y)t$ being the rescaled time such that $t_S = 1$ when $t = t_\trm{max} \equiv (M + \gamma_y)^{-1}$. Plots (a), (c), and (d) (left column)  have been generated with a field fluctuation strength of $q_B = 100G^2/s$, and plots (b), (d), (f) (right column) with $q_B = 10^4G^2/s$. The first row (plots (a) and (b)) show the fluctuating field in solid blue, juxtaposed with the confidence interval of $\overline{B}_t$, $\pm 2\sqrt{\text{Var}[\, \oline{B}_t]}$, as well as the upper bound on $|\overline{B}_t|\lesssim\sqrt{2}(M + \gamma_y)/\gmr$ for which equation \eqref{eq:ap_UncondJx_appr1} can be approximated by \eqref{eq:ap_UncondJx_appr}. The plots in the second row (subfigures (c) and (d)), compare the exact solution of $\braket{\hat{J}_x(t)}$ with its approximation \eqref{eq:ap_UncondJx_appr}, $J e^{-(M + \gamma_y)t/2}$. Finally, in the bottom row, plots (e) and (f) show the error percentage for this approximation of $\braket{\hat{J}_x(t)}$.
  }
  \label{fig:AppendixA_Figure}
\end{figure}

Although when constructing the optimal estimator (KF) of $B_t$ at finite time $t>0$ we do not want to assume any prior knowledge about the value of the initial field $B_0$, for the purpose of this calculation we take the magnetic field to always be initialised in $B_0=0$, so that it is solely the presence of the OU process \eqref{eq:OU_process} that invalidates the condition \eqref{eq:t_UB_AvB_t} over time. In such a case, we have
\begin{align}
\left|\mean{\oline{B}_t} \pm 2\sqrt{\text{Var}[\, \oline{B}_t]}\right| = \left|2\sqrt{\mean{(\,\oline{B}_t)^2}}\right|
& = \sqrt{\frac{2q_B}{\chi}}\sqrt{\frac{4e^{-\chi t} + 2\chi t -e^{-2\chi t} - 3}{\chi^2 t^2}} \\
& \approx \sqrt{\frac{2q_B}{\chi}}\sqrt{\frac{2 \chi t}{3}}=\sqrt{\frac{4q_Bt}{3}}, \label{eq:AvB_t_max}
\end{align}
where the above approximation holds as long as $t \ll \frac{4}{3\chi}$. Hence, substituting \eqref{eq:AvB_t_max} into \eqref{eq:t_UB_AvB_t}, we obtain the desired condition setting an upper limit on valid timescales as a function of the fluctuations strength $q_B$, i.e:~$t \lesssim \frac{1}{2\gmr}\sqrt{\frac{3}{q_Bt}} \implies t \lesssim \left(\frac{3}{4\gmr^2 q_B}\right)^{1/3}$.

In summary, we conclude that the linear-Gaussian approximation \eqref{eq:ap_UncondJx_appr} is valid at timescales short enough such that all the three conditions $t\lesssim(M+\gamma_y)^{-1}$, $t \ll \frac{4}{3\chi}$ and $t \lesssim \left(\frac{3}{4\gmr^2 q_B}\right)^{1/3}$ hold. Throughout our work we assure these by focussing on the first one and considering the rescaled time $t_S=(M+\gamma_y)t$ to always fulfil $t_S\lesssim1$ (as also done in \fref{fig:AppendixA_Figure}). However, we must then also require the field fluctuations to be small enough such that the other two conditions are satisfied for any $t_S\lesssim1$. In particular, this is achieved by considering the decay parameter and the fluctuation strength of the OU process \eqref{eq:OU_process} to fulfil for any $t\lesssim(M+\gamma_y)^{-1}$:
\begin{align}
\chi \ll \frac{4}{3t} &\qquad\implies\qquad \chi \ll\frac{4(M+\gamma_y)}{3}
\label{eq:chi_constraint}\\
q_B \lesssim \frac{3}{4\gmr^2 t^3} &\qquad\implies\qquad q_B \lesssim \frac{3(M+\gamma_y)^3}{4\gmr^2},
\label{eq:qB_constraint}	
\end{align}
which we more generally state for $r=M+\gamma_y+\gamma_z$ in equation \eqref{eq:OU_process_constraints} of the main text.

The importance of the upper constraint on the strength of field fluctuations $q_B$ we demonstrate in \fref{fig:AppendixA_Figure}, where we present two exemplary Wiener ($\chi=0$) trajectories of $B_t$ with $q_B=10^2$ and $q_B=10^4$ for $3(M+\gamma_y)^3/(4\gmr^2)\approx10^3$ (all in $[G^2/s]$), such that in the latter case the condition \eqref{eq:qB_constraint} is  clearly invalidated. Correspondingly, as directly seen from \fref{fig:AppendixA_Figure}b, the range of valid timescales \eqref{eq:t_UB_AvB_t} is surpassed around $t_S=0.4$, at which the linear-Gaussian approximation \eqref{eq:ap_UncondJx_appr} also ceases to hold, as shown in \fref{fig:AppendixA_Figure}d. In contrast, for $q_B=10^2$ the approximation \eqref{eq:ap_UncondJx_appr} is consistently maintained with accuracy within $0.25\%$ for any $t_S\lesssim1$---see \fref{fig:AppendixA_Figure}e.

\section{Conditional dynamics of the variance $\Delta^2 \hat{J}_z$}
\label{ap:VarSol}
The differential equation for the conditional variance of $\hat{J}_z$ introduced in \eqref{eq:var_eq_Jz_before} reads
\begin{align} \label{eq:ap_Vardif}
     d\braket{\Delta^2 \hat{J}_z (t)}_\cc & = -4 M \eta \braket{\Delta^2 J_Z(t)}^2_\cc dt + \gamma_y J^2 e^{-(M + \gamma_y) t} dt,
\end{align}
whose solution can be expressed in terms of modified Bessel functions of first and second kind, $\mathcal{I}_{\beta}[\tinyspace \cdot \tinyspace]$ and $\mathcal{K}_{\beta}[\tinyspace \cdot \tinyspace]$, and regularized confluent hypergeometric functions $\prescript{}{0}{F}_1[\tinyspace \cdot \tinyspace]$, i.e.:
\begin{align} \label{eq:exactsol}
      \braket{\Delta^2 \hat{J}_z (t)}_\cc = V_e(t) & = J e^{-(M+\gamma_y)t/2}  \Bigg(\mathcal{I}_{1}\Big[2\beta\Big] \Big(\sqrt{\eta \, \gamma_y M} \tinyspace \mathcal{K}_0[2\alpha] - \gamma_y \tinyspace \mathcal{K}_1[2\alpha]\Big) + \nonumber  \\
    &  + \mathcal{K}_1\Big[2\beta\Big] \Big(\gamma_y \, \mathcal{I}_1[2\alpha] + \sqrt{\eta \, \gamma_y M}  \, \prescript{}{0}{F}_1[1,\alpha^2] \Big) \Bigg)  \nonumber  \\
    &  \Bigg/ \Bigg( 2\prescript{}{0}{F}_1[1,\beta^2] \Big(\sqrt{\eta \, \gamma_y M} \mathcal{K}_1[2\alpha] - M\eta \, \mathcal{K}_0[2\alpha]\Big) \nonumber  \\
    &  +\frac{2 \eta M}{M + \gamma_y} \mathcal{K}_0[2\beta] \Big((M+\gamma_y) \, \prescript{}{0}{F}_1[1,\alpha^2] + 2\gamma_y J  \, \prescript{}{0}{F}_1[2,\alpha^2]\Big) \Bigg),
\end{align}
where $\alpha = 2J\sqrt{\eta \, \gamma_y M}/(M + \gamma_y)$ and $\beta = \alpha \tinyspace e^{-(M + \gamma_y)t/2}$. The behaviour of the solution \eqref{eq:exactsol} can be better understood when broken down into different regimes.

In order to do so, the first step is to expand the modified Bessel functions and the regularized confluent hypergeometric functions around infinity up to leading order---such an approximation is assured due to $\alpha \gg 1$ and $\beta \gg 1$. The relevant expansions are shown in \tref{tab:Bessel_exp}.
\begin{table}[ht]
\centering
\begin{tabular}{|p{4cm}|p{4cm}|p{4cm}|}
\hline
\begin{equation*}
\mathcal{I}_1[2\beta] \approx \frac{e^{2\beta}}{2\sqrt{\pi \beta}}
\end{equation*} &
\begin{equation*}
\mathcal{K}_0[2\alpha] \approx \frac{1}{2}\sqrt{\frac{\pi}{\alpha}}e^{-2\alpha}
\end{equation*} &
\begin{equation*}
\mathcal{K}_1[2\alpha] \approx  \frac{1}{2}\sqrt{\frac{\pi}{\alpha}}e^{-2\alpha}
\end{equation*}\\
\hline
\begin{equation*}
\mathcal{K}_1[2\beta] \approx  \frac{1}{2}\sqrt{\frac{\pi}{\beta}}e^{-2\beta}
\end{equation*} &
\begin{equation*}
\mathcal{I}_1[2\alpha] \approx \frac{e^{2\alpha}}{2\sqrt{\pi \alpha}}
\end{equation*} &
\begin{equation*}
\prescript{}{0}{F}_1[1,\alpha^2] \approx \frac{e^{2\alpha}}{2\sqrt{\pi \alpha}}
\end{equation*}\\
\hline
\begin{equation*}
\prescript{}{0}{F}_1[1,\beta^2] \approx \frac{e^{2\beta}}{2\sqrt{\pi \beta}}
\end{equation*} &
\begin{equation*}
\mathcal{K}_0[2\beta] \approx \frac{1}{2} \sqrt{\frac{\pi}{\beta}} e^{-2\beta}
\end{equation*} &
\begin{equation*}
\prescript{}{0}{F}_1[2,\alpha^2] \approx \frac{e^{2\alpha}}{2 \alpha \sqrt{\pi \alpha}}
\end{equation*}\\
\hline
\end{tabular}
\caption{Series expansions of the Bessel functions for $1/\alpha$ and $1/\beta$ around $1/\alpha_0 = 0$ and $1/\beta_0 = 0$, stated to leading order.}
\label{tab:Bessel_exp}
\end{table}

Substituting the leading-order expansions of \tref{tab:Bessel_exp} into the solution \eqref{eq:exactsol} and approximating $e^{4\beta}\approx e^{-2\alpha(-2+t(M+\gamma_y))}$, the variance of $\hat{J}_z(t)$ simplifies to
\begin{align}\label{eq:approx_varJz}
    \braket{\Delta^2 \hat{J}_z (t)}_\cc \approx \frac{1}{2}Je^{-(M+\gamma_y)t/2} \frac{\sqrt{M\gamma_y\eta} \cosh{(2Jt\sqrt{M\gamma_y \eta})}+\gamma_y \sinh{(2Jt\sqrt{M\gamma_y \eta})}}{\sqrt{M\gamma_y \eta} \cosh{(2Jt \sqrt{M\gamma_y\eta})}+M\eta \sinh{(2Jt\sqrt{M\gamma_y\eta})}}.
\end{align}
Note that if $2Jt\sqrt{M\gamma_y \eta} \gg 1$ then $\cosh{(2Jt\sqrt{M\gamma_y \eta})}  \approx \frac{1}{2} e^{2Jt\sqrt{M\gamma_y \eta}}$ and $\sinh{(2Jt\sqrt{M\gamma_y \eta})} \approx \frac{1}{2} e^{2Jt\sqrt{M\gamma_y \eta}}$. As a result, expression \eqref{eq:approx_varJz} simplifies further to the form stated in equation \eqref{eq:varJz_t>>t*}, i.e.:
\begin{align}\label{eq:approx_varJz_1}
    \braket{\Delta^2 \hat{J}_z (t)}_\cc \approx V_{>t^*}(t) = \frac{1}{2}Je^{-(M+\gamma_y)t/2} \sqrt{\frac{\gamma_y}{\eta M}}.
\end{align}
Analogously, if $2Jt\sqrt{M\gamma_y \eta} \ll 1$, then $\cosh{(2Jt\sqrt{M\gamma_y \eta})} \approx 1$ and $\sinh{(2Jt\sqrt{M\gamma_y \eta})} \approx 2Jt\sqrt{M\gamma_y \eta}$, and the expression \eqref{eq:approx_varJz} simplifies to the form stated in equation \eqref{eq:varJz_t<<t*}:
\begin{align}\label{eq:approx_varJz_2}
    \braket{\Delta^2 \hat{J}_z (t)}_\cc \approx V_{<t^*}(t) = Je^{-(M+\gamma_y)t/2} \, \frac{(1+2 J t \gamma_y)}{2+4 J t M \eta}.
\end{align}
Moreover, since $2Jt\sqrt{M\gamma_y \eta} \ll 1$, we can then derive the condition
\begin{align}
  t \ll t^{*} = \frac{1}{2J\sqrt{M\gamma_y \eta}}
  \quad\implies\quad
  2 J t \gamma_y \ll \sqrt{\frac{\gamma_y}{M \eta}},
\end{align}
for which approximation \eqref{eq:approx_varJz_2} holds.

From equation \eqref{eq:approx_varJz_2} it also follows that \emph{if} $\gamma_y < M$, then $(1+2 J t \gamma_y) \approx 1$ and, hence, for sufficiently short times scales $t \ll t^{*} < (M + \gamma_y)^{-1}$ we can always write
\begin{equation}
  \braket{\Delta^2 \hat{J}_z (t)}_\cc \approx \frac{J}{2+4 J t M \eta} e^{-(M+\gamma_y)t/2} \approx \frac{J}{2+4 J t M \eta}, \label{eq:ap_Geremia_sol}
\end{equation}
which is the noiseless solution derived originally in~\cite{Geremia2003}. In the other direction, by showing that $\braket{\Delta^2 \hat{J}_z (t)}_\cc$ is a non-decreasing function at $t=0_+$ if $\gamma_y \ge \eta M$, we can prove that \eqref{eq:ap_Geremia_sol} holds \emph{only if} $\gamma_y < \eta M$ for $J\gg1$. Namely, that we can consider the global decoherence $\gamma_y$ to be insignificant at small times $t \ll t^*$ only when $\gamma_y < \eta M$. In order to do so, we differentiate the expression \eqref{eq:approx_varJz_2} w.r.t.~$t$ and then let $t\to0$, i.e.:
\begin{equation}
  \lim_{t \rightarrow 0}\frac{d}{dt} \left[V_{<t^*}(t)\right] = -\frac{1}{4} J \,[\gamma_y(1-4J)+M(1+4 J \eta)].
\end{equation}
Setting the above expression to zero, we find the value of $\gamma_y$ for which the derivative changes signs at $t=0$, i.e.:
\begin{equation}
  \gamma_y = \frac{M + 4 J M \eta}{4J - 1} = M \eta + \frac{M(\eta + 1)}{4J} + \frac{M(\eta + 1)}{16 J^2} + O\!\left(\frac{1}{J^3}\right),
\end{equation}
which can be correctly approximated as $\gamma_y = \eta M$ when $J \gg 1$. Hence,
\begin{equation}
  \lim_{t \rightarrow 0} \frac{d}{dt} \left[ V_{<t^*} (t) \right] \geq 0 \quad\text{if}\quad \gamma_y \geq \eta M,
  \qquad\text{and}\qquad
  \lim_{t \rightarrow 0} \frac{d}{dt} \left[ V_{<t^*} (t) \right] < 0 \quad\text{if}\quad \gamma_y < \eta M,
\end{equation}
what proves that $\braket{\Delta^2 \hat{J}_z (t)}_\cc$ can be approximated as in \eqref{eq:ap_Geremia_sol} at short enough timescales \emph{only} when $\gamma_y < \eta M$.

Finally, in \fref{fig:ap_Variance_plots} we compare the exact solution for the variance given in \eqref{eq:exactsol} with the short/long-time approximations $ V_{<t^*} (t)$ and $ V_{>t^*} (t)$ shown in \eqref{eq:approx_varJz_2} and \eqref{eq:approx_varJz_1}, respectively.

\begin{figure}[t]
  \centering
  \includegraphics[width=\linewidth]{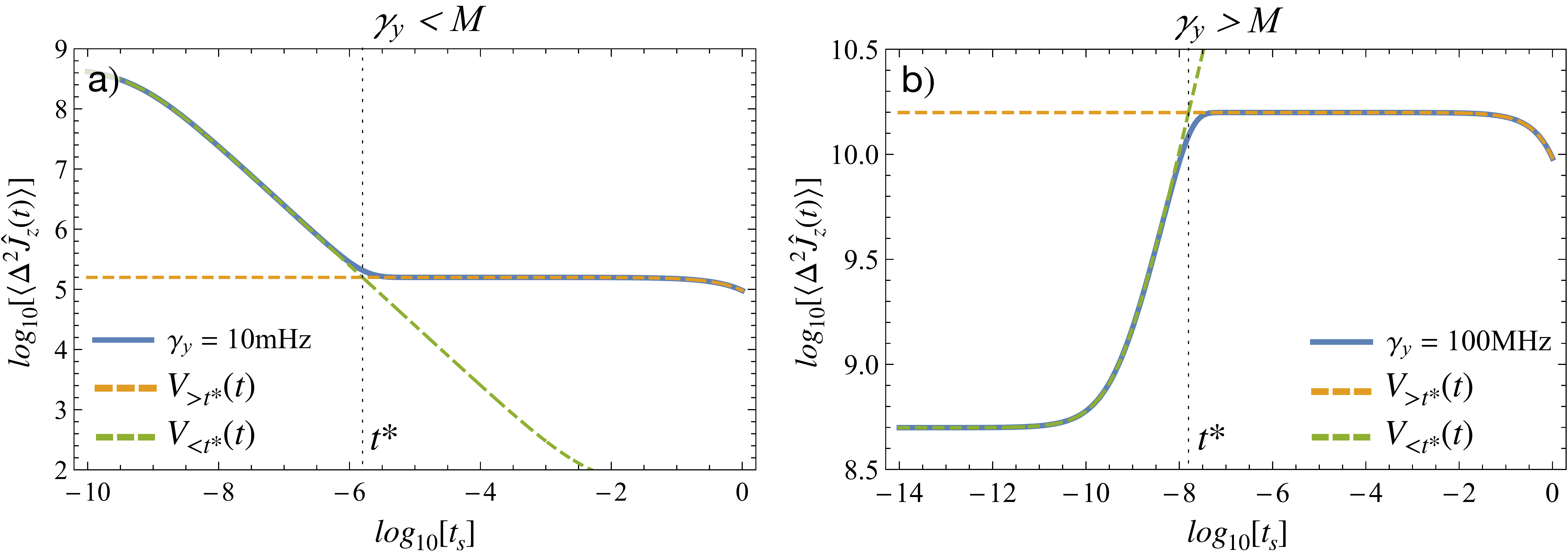}
  \caption{In subfigure (a) with $\gamma_y = 10mHz < M$, the exact variance solution $\braket{\Delta^2 \hat{J}_z (t)}_\cc$ is compared to the approximated functions $V_{<t^{*}}(t)$ and $V_{>t^{*}}(t)$ (dashed green and yellow, respectively). The transition time $t^*$ between these two regimes is marked with a dotted black vertical line. In subfigure (b) with $\gamma_y = 100 MHz > M$, the two different regimes $V_{<t^*}(t)$ (dashed green) and $V_{>t^*}(t)$ (dashed yellow) are superimposed with the exact solution of $\braket{\Delta^2 \hat{J}_z (t)}_\cc$ (in solid blue). Notation $t_S$ refers to a rescaled time, $t_S = t (M + \gamma_y)$. All plots have been generated with $M = 100$kHz, $\gmr = 1$kHz/mG, $\eta = 1$, and $J = 10^9$.}
  \label{fig:ap_Variance_plots}
\end{figure}

\section{Construction of the Kalman filter for correlated dynamics}
\label{ap:Corr_KF}
Consider the following system of stochastic differential equations within the It\^{o} calculus:
\begin{align}
    \label{eq:x} d\vecvar{x}_t & = \ \textbf{A}(\vecvar{x}_t,t) dt + \mat{B}(\vecvar{x}_t,t) d\vecvar{w}_t, \\
    \label{eq:y} d\vecvar{y}_t & = \ \textbf{C}(\vecvar{x}_t,t) dt + d\vecvar{v}_t,
\end{align}
where $d\vecvar{w}_t$ and $d\vecvar{v}_t$ are process and measurement Wiener noise-terms, respectively, that exhibit zero mean $\mean{d\vecvar{w}_t} = \mean{d\vecvar{v}_t} = 0$ and self-correlations:
\begin{align}
    \mean{d\vecvar{w}_t d\vecvar{w}_s^\Trans} & = \ \mat{Q}_t \tinyspace \delta(t-s) dt, \label{eq:dw}\\
    \mean{d\vecvar{v}_t d\vecvar{v}_s^\Trans} & = \ \mat{R}_t \tinyspace \delta(t-s) dt \label{eq:dv},
\end{align}
defined with help of (positive semi-definite) covariances $\mat{Q}_t,\mat{R}_t\ge0$.
Moreover, the cross-correlations between $d\vecvar{w}_t$ and $d\vecvar{v}_t$ are specified by the matrix $\mat{S}_t$:
\begin{align}
    \mean{d\vecvar{w}_t d\vecvar{v}_t^\Trans} = \mat{S}_t \tinyspace \delta(t-s) \tinyspace dt,
\end{align}
which is not necessarily symmetric. Note that the matrices $\mat{Q}_t$, $\mat{R}_t$ and $\mat{S}_t$ are not fully independent: since the covariance of the ``overall'' noise $d\textbf{T} = d\vecvar{w}_t \oplus d\vecvar{v}_t$, i.e.
\begin{align} \label{eq:ap_condition}
    \mean{d\mat{T}_t d\mat{T}_t^\Trans} = \left[\begin{array}{c|c} \mat{Q}_{t} & \mat{S}_{t}\\
    \hline \mat{S}_{t}^{T} & \mat{R}_{t} \end{array}\right]\ge0,
\end{align}
must be positive semi-definite by definition (constituting an outer product of a vector), the form that matrices $\mat{Q}_t$, $\mat{R}_t$ and $\mat{S}_t$ can take is then constrained.

Next, it is convenient to rewrite the system of differential equations \eqrefs{eq:x}{eq:y} as~\cite{crassidis2011optimal}:
\begin{align} \label{eq:dx_eq}
    d\vecvar{x}_t & = \ \mat{A}(\vecvar{x}_t,t) \ dt + \mat{B}(\vecvar{x}_t,t) \ d\vecvar{w}_t + \mat{D}(\vecvar{x}_t,t) \big(d\vecvar{y}_t - \mat{C}(\vecvar{x}_t,t) dt - d\vecvar{v}_t \big),
\end{align}
where $D(\vecvar{x}_t,t)$ can be set arbitrarily, since $d\vecvar{y}_t - \textbf{C}(\vecvar{x}_t,t) dt - d\vecvar{v}_t = 0$. By regrouping the terms in equation \eqref{eq:dx_eq} we obtain
\begin{align}
    d\vecvar{x}_t & = \big( \mat{A}(\vecvar{x}_t,t) - \mat{D}(\vecvar{x}_t,t) \mat{C}(\vecvar{x}_t,t) \big) dt + \mat{D}(\vecvar{x}_t,t) \ d\vecvar{y}_t + d\mat{Z}_t,
\end{align}
with $d\mat{Z}_t = \mat{B}(\vecvar{x}_t,t) d\vecvar{w}_t - \mat{D}(\vecvar{x}_t,t) d\vecvar{v}_t$ being now the effective process noise. This allows us to ensure that the correlations between the process and measurement noises, $d\mat{Z}_t$ and $d\vecvar{v}_t$, respectively, are vanishing---and after setting $\mat{D}(\vecvar{x}_t,t) = \mat{B}(\vecvar{x}_t,t)\mat{S}_t\mat{R}^{-1}_t$ without loss of generality we have
\begin{align}
    \mean{d\mat{Z}_t d\vecvar{v}_t^\Trans} = \big( \mat{B}(\vecvar{x}_t,t)\mat{S}_t - \mat{D}(\vecvar{x}_t,t)\mat{R}_t \big) \, dt = 0.
\end{align}

As a result, we obtain a set of stochastic differential equations that are equivalent to \eqrefs{eq:x}{eq:y}:
\begin{align}
    d\vecvar{x}_t & = \ \mat{A}(\vecvar{x}_t,t) dt + \mat{B}(\vecvar{x}_t,t)\mat{S}_t\mat{R}^{-1}_t\big(d\vecvar{y}_t-\textbf{C}(\vecvar{x}_t,t)dt\big) + \mat{B}(\vecvar{x}_t,t)d\mat{U}_t, \label{eq:x_uncorr}\\
    d\vecvar{y}_t & = \ \mat{C}(\vecvar{x}_t,t) dt + d\vecvar{v}_t \label{eq:y_uncorr},
\end{align}
where the process noise $d\mat{U}_t = \mat{B}(\vecvar{x}_t,t)^{-1} d\mat{Z}_t = d\vecvar{w}_t - \mat{S}_t\mat{R}^{-1}_t d\vecvar{v}_t$ is now uncorrelated from the measurement noise, $\mean{d\mat{U}_t d\vecvar{v}^\Trans_t} = 0$, but the dynamics of the state $\vecvar{x}_t$ explicitly depend on the observation $\vecvar{y}_t$ in \eqref{eq:y_uncorr}. Moreover, in contrast to \eqrefs{eq:dw}{eq:dv}, the self-correlations of process and measurement noises now read
\begin{align}
    &\mean{d\mat{U}_t d\mat{U}_t^\Trans} = \big(\mat{Q}_t - \mat{S}_t\mat{R}^{-1}_t\mat{S}^\Trans_t \big)dt, \\
    &\mean{d\vecvar{v}_t d\vecvar{v}^\Trans_t} = \mat{R}_t dt.
\end{align}
Now, in case of a linear model, the dynamical matrices in \eqrefs{eq:x}{eq:y} simplify to:
\begin{align}
    &\textbf{A}(\vecvar{x}_t,t) = \mat{F}_t \tinyspace \vecvar{x}_t, \\
    &\mat{B}(\vecvar{x}_t,t) = \mat{B}_t, \\
    &\textbf{C}(\vecvar{x}_t,t) = \mat{H}_t \tinyspace \vecvar{x}_t,
\end{align}
so that the system of equations \eqrefs{eq:x_uncorr}{eq:y_uncorr} reads
\begin{align}
    \label{eq:dif_eq_corr_1} d\vecvar{x}_t & = \ \mat{F}_t \vecvar{x}_t \tinyspace dt + \mat{B}_t \tinyspace \mat{S}_t \tinyspace \mat{R}^{-1}_t \big(d\vecvar{y}_t-\mat{H}_t \tinyspace \vecvar{x}_t dt\big) + \mat{B}_t d\textbf{U}_t, \\
    \label{eq:dif_eq_corr_2} d\vecvar{y}_t & = \ \mat{H}_t \tinyspace \vecvar{x}_t \tinyspace dt + d\vecvar{v}_t.
\end{align}
In such a special case, the optimal estimator $\Tilde{\vecvar{x}}_t$ minimising the overall aMSE, which is defined as the trace of the covariance matrix $\mat{\Sigma}_t = \mean{(\vecvar{x}_t - \Tilde{\vecvar{x}}_t)(\vecvar{x}_t - \Tilde{\vecvar{x}}_t)^\Trans}$, i.e.~$\tr{\mat{\Sigma}_t} = \mean{\left\Vert\vecvar{x}_t - \Tilde{\vecvar{x}}_t\right\Vert^2}$, is referred to as the \emph{Kalman filter} (KF) and given by the solution of the so-called Kalman-Bucy equation~\cite{crassidis2011optimal}:
\begin{align} \label{eq:ap_KBF}
    d\Tilde{\vecvar{x}}_t = \mat{F}_t \Tilde{\vecvar{x}}_t \tinyspace dt + \mat{\Gamma}_t \big(d\vecvar{y}_t - \mat{H}_t \tinyspace \Tilde{\vecvar{x}}_t \tinyspace dt\big)
\end{align}
where $\mat{\Gamma}_t$ is the Kalman gain
\begin{align}
    \mat{\Gamma}_t = \big(\mat{\Sigma}_t \tinyspace \mat{H}^\Trans_t + \mat{B}_t \tinyspace \mat{S}_t\big)\mat{R}^{-1}_t.
    \label{eq:ap_KG_general}
\end{align}

Nonetheless, let us note that in order to evaluate the Kalman gain \eqref{eq:ap_KG_general}, one must determine beforehand the (optimal) covariance matrix $\mat{\Sigma}_t$, which evolves according to the non-linear differential equation in the Riccati form~\cite{crassidis2011optimal}:
\begin{align}
	\frac{d\mat{\Sigma}_t}{dt} \;
	& = \; \mat{F}_t \mat{\Sigma}_t + \mat{\Sigma}_t \tinyspace \mat{F}^\Trans_t - \mat{\Gamma}_t \mat{R}_t \mat{\Gamma}^\Trans_t + \mat{B}_t \mat{Q}_t \mat{B}^\Trans_t \\
	& = \; \Big(\mat{F}_t - \ \mat{B}_t \mat{S}_t \mat{R}^{-1}_t \mat{H}_t \Big) \mat{\Sigma}_t + \mat{\Sigma}_t \Big(\mat{F}_t - \ \mat{B}_t \mat{S}_t \mat{R}^{-1}_t \mat{H}_t \Big)^\Trans \nonumber \\
    & \qquad - \mat{\Sigma}_t \mat{H}^\Trans_t \mat{R}^{-1}_t \mat{H}_t \mat{\Sigma}_t + \mat{B}_t \Big( \mat{Q}_t - \mat{S}_t \mat{R}^{-1}_t \mat{S}^\Trans_t \Big) \mat{B}^\Trans_t.
	\label{eq:ap_Riccati}
\end{align}
From the practical perspective, however, the solution to equation \eqref{eq:ap_Riccati} can be determined (often only numerically) and stored in advance, so that in real-life applications the construction of the KF, $\Tilde{\vecvar{x}}_t$, as the solution to equation \eqref{eq:ap_KBF} can still performed fast and ``on the fly'', while the observations $\vecvar{y}_t$ are constantly gathered.

\section{Steady-state solution of the Kalman filter for $\chi \neq 0$}
\label{ap:SSS}
Let us consider the Riccati differential equation \eqref{eq:ap_Riccati} (equivalent to \eqref{eq:Riccati_equation}), which specifies the evolution of the covariance matrix $\mat{\Sigma}_t$ for the KF with the dynamical matrices $\mat{F}_t$, $\mat{B}_t$, and $\mat{H}_t$, as well as the noise-correlation matrices $\mat{Q}_t$, $\mat{R}_t$, and $\mat{S}_t$ defined in equations \eqref{eq:matrices_F_K_B} and \eqref{eq:noise_correlations} of the main text, respectively.

For simplicity, we rename the elements of the covariance matrix as
\begin{equation}
	\mat{\Sigma}_t = \begin{pmatrix} x(t) & y(t) \\ y(t) & z(t) \end{pmatrix} \ge 0 ,
\end{equation}
so that by resorting to the evolution of $\mat{\Sigma}_t$ in \eqref{eq:ap_Riccati} and setting $\frac{d\mat{\Sigma}}{dt}=0$---with \eqref{eq:ap_Riccati} formally constituting then a continuous-time algebraic Riccati equation (CARE)~\cite{crassidis2011optimal}---we obtain a system of equations describing the steady state as
\begin{align}
    -8M\eta V_{>t^*}(t) \, x(t) - 4M\eta \, x^2(t) - 2 \gmr J e^{-(M + \gamma_y) t/2} \, y(t) = 0, \nonumber \\
    -\chi \, y(t) - 4 M \eta V_{>t^*}(t) \, y(t) - 4 M\eta x(t) y(t) - \gmr J e^{-(M+\gamma_y)t/2} z(t) = 0, \nonumber \\
    q_B - 4M\eta y^2(t) -2 \chi z(t) = 0,
    \label{eq:CARE}
\end{align}
where we have used the fact that the variance $\braket{\Delta^2 \hat{J}_z(t)}_\cc$ for $t \gg t^*$ (and, hence, in the steady state) equals $V_{>t^*}(t)$ as defined in equation \eqref{eq:varJz_t>>t*}.

Being interested only in the steady-state (SS) solution for the aMSE of the magnetic field, $\Delta^2\est{B}_t^\text{SS} \equiv z(t)$, one may explicitly solve for $z(t)$ in \eqref{eq:CARE} to obtain
\begin{align} \label{eq:VarB}
    \Delta^2\est{B}_t^\text{SS} & = -\frac{\gamma_y \chi}{\gmr^2} - \frac{\chi^3}{4 \gmr^2 J^2 M \eta} e^{(M + \gamma_y)t} - \frac{\chi}{\gmr^2 J \sqrt{M \eta}} \sqrt{q_B \gmr^2 + \gamma_y \chi^2} \, e^{(M + \gamma_y)t/2} \nonumber \\
    & \quad + \frac{1}{2\gmr^2 J M \eta} \sqrt{q_B \gmr^2 + \gamma_y \chi^2}  \left(\sqrt{M\eta} + \frac{\chi^2  e^{(M+\gamma_y)t/2}}{2 J \sqrt{q_B \gmr^2 + \gamma_y \chi^2}} \right) \times \\
    & \quad\qquad\times \sqrt{\chi^2 e^{(M+\gamma_y)t} + 4 J \left(\gamma_y J M \eta + e^{(M+\gamma_y)t/2}\sqrt{M \eta(q_B \gmr^2 + \gamma_y \chi^2)}\right)}, \nonumber
\end{align}
which for $J\gg1$ can be approximated by performing the Taylor expansion in $1/J$ to first order, as follows:
\begin{align} \label{eq:approx_VarB}
    \Delta^2\est{B}_t^\text{SS}
    \;\underset{J\gg1}{\approx}\;
    -\frac{\gamma_y \chi}{\gmr^2}+\frac{\gamma_y}{\gmr^2} \sqrt{\frac{q_B \gmr^2 + \gamma_y \chi^2}{\gamma_y}}.
\end{align}
In contrast, by letting $\chi \to 0$ in equation \eqref{eq:VarB}, we obtain the exact steady-state solution for the aMSE, when the magnetic field follows a Wiener (rather than the OU \eqref{eq:OU_process}) process:
\begin{align}
    \Delta^2\est{B}_t^\text{SS}|_{\chi = 0} = \left( \frac{q_B \, \gamma_y}{\gmr^2} + \frac{1}{\gmr J} \sqrt{\frac{q_B^3}{M \eta}} e^{(M+\gamma_y)t/2} \right)^{1/2},
\end{align}
which is the expression stated in equation \eqref{eq:ss_field_sol} of the main text.

\section{Derivation of the CS limit for a fluctuating field}
\label{ap:BCRB}
After discretising the dynamics of the magnetic field $B_t$ and the measurement outcomes $y_t$ into $k=t/\delta t$ steps, their particular trajectories follow discrete-time stochastic processes being described by (time-ordered) sets:
\begin{align}
    \vec{B}_{k} = \{B_0,B_1,\dots,B_k\}
    \qquad\text{and}\qquad
    \vec{y}_{k} = \{y_0,y_1,\dots,y_k\},
\end{align}
respectively, while the marginal conditional BCRB \eqref{eq:BCRB} after the $k$th step then reads
\begin{equation}\label{eq:ap_BCRB}
  \Delta^2\est{B}_k \geq J_B^{-1} = (J_P + J_M)^{-1} = \frac{1}{\F[p(B_k)] + \int dB_k \; p(B_k) \, \F[p(\vec{y}_{k}|B_k)]}.
\end{equation}

In the first subsection of this appendix, we compute the prior contribution to the BI, i.e.~$J_P$ above in \eqref{eq:ap_BCRB}. In the second subsection, we show how the quantum channel describing the dynamics of the atomic ensemble in the absence of continuous measurement can be decomposed into a probabilistic mixture of unitary channels. Finally, in the last subsection, we calculate the FI of $\mathbb{P}_{B_k}(\vec{\omega}_{k})$ defined in \eqref{eq:bigP_def}, which allows us to upper-bound the contribution of the measurement records to the BI, i.e.~$J_M$ above in \eqref{eq:ap_BCRB}.

\subsection{Prior contribution to the Bayesian information}
\label{ap:prior_distribution}
The marginal probability density function of the magnetic field at time $t$---or equivalently after the time-step $k=t/\delta t$---can be written as:
\begin{align} \label{eq:prior1}
    p(B_k) = \int\! d\vec{B}_{k-1}\;p(\vec{B}_{k}) = \int\! d\vec{B}_{k-1}\;\prod_{j = 1}^k p(B_j|B_{j-1})\;p(B_0),
\end{align}
where the initial field $B_0$ is drawn from a Gaussian distribution which effectively specifies our \textit{a priori} knowledge about the field at $t=0$, i.e.:
\begin{align} \label{eq:B0_dist}
    B_0 \sim p(B_0) = \frac{1}{\sqrt{2\pi \sigma_0^2}}\, \exp\!\left(-\frac{B_0^2}{2\sigma_0^2}\right),
\end{align}
while each transition probability $p(B_j|B_{j-1})$ for all $j=1,2,\dots,k$ is given by the OU process \eqref{eq:OU_process} as a Gaussian distribution~\cite{Gardiner1985}:
\begin{align} \label{eq:OUP_dist}
    p(B_j|B_{j-1}) = \sqrt{\frac{1}{2 \pi V_P }}\, \exp\!\left(-\frac{(B_j - B_{j-1}e^{-\chi \delta t})^2}{2V_P}\right)
\end{align}
with variance
\begin{equation}\label{eq:ap_OUPvar}
  V_{P} = \frac{q_B}{2\chi} (1 - e^{-2 \chi \delta t}).
\end{equation}

Hence, after explicitly evaluating the multiple integrals in \eqref{eq:prior1}, we arrive at the marginal Gaussian distribution for $B_k$:
\begin{align} \label{eq:ap_totalpBk_OUP}
    p(B_k) = \sqrt{\frac{1}{2\pi V_P^{(k)}}}\, \exp\!\left(-\frac{B_k^2}{2V_P^{(k)}}\right),
\end{align}
whose mean is zero and its variance
\begin{equation}\label{eq:ap_totVarP}
  V_P^{(k)} = \sigma_0^2 e^{-2k\chi\delta t} + \frac{q_B}{2\chi}(1 - e^{-2k\chi\delta t}),
\end{equation}
which in the case of a magnetic field following a Wiener process---a special case of the OU process \eqref{eq:OU_process} with $\chi=0$---simplifies to $\lim_{\chi\to0}V_P^{(k)}=k \delta t q_B + \sigma_0^2$.

In general, the FI evaluated with respect to a Gaussian distribution, e.g.~the marginal distribution \eqref{eq:ap_totalpBk_OUP}, corresponds to the inverse of its variance. Hence, the prior contribution to the BI defined in equation \eqref{eq:J_P} of the main text simply reads
\begin{align} \label{eq:J_P_app}
    J_P = \F[p(B_k)] = \frac{1}{V_P^{(k)}} = \frac{1}{ \sigma_0^2 e^{-2k\chi\delta t} + \frac{q_B}{2\chi} (1 - e^{-2k\chi\delta t})},
\end{align}
and in the continuous-time limit of $\delta t \rightarrow 0$ with $k = t/\delta t$, it becomes
\begin{equation}\label{eq:JpCont}
  J_P = \frac{1}{ \sigma_0^2 e^{-2 \chi t} + \frac{q_B}{2\chi} (1 - e^{-2 \chi t})}.
\end{equation}
In the special case of a Wiener process ($\chi \rightarrow 0$), equation \eqref{eq:JpCont} reduces to $\lim_{\chi\to0} J_P = (q_B \, t + \sigma_0^2)^{-1}$. On the other hand, for any $\chi,q_B,t\ge0$, expression \eqref{eq:JpCont} vanishes if we do not possess any prior knowledge about the initial field $B_0$, i.e.~when we consider $\sigma_0 \rightarrow \infty$ in equation \eqref{eq:B0_dist}, for which $\lim_{\sigma_0\to\infty} J_P=0$.

\subsection{Noisy dynamics as a convex mixture of unitary channels}
\label{sec:CS_decomposition}
Consider a unitary evolution, which is generated by $\xi \hat{H}$ with $\hat{H}$ being a time-independent Hamiltonian and $\xi\in\mathbb{R}$, being applied to a state $\rho_{0}$ for a time interval $\tau$, i.e.:
\begin{equation}
    \Unitary_{\xi, \delta t}[\rho_{0}] = e^{-i \xi \hat{H} \tau} \rho_{0} e^{i \xi \hat{H} \tau}
\end{equation}
where the frequency-like parameter $\xi$ is randomly distributed according to a Gaussian probability density:
\begin{equation}
    \xi \sim p_{\mu_\tau,\sigma_\tau}(\xi) = \frac{1}{\sqrt{2 \pi \sigma_\tau^2}} e^{-\frac{(\xi - \mu_\tau)^2}{2 \sigma_\tau^2}},
\label{eq:Gaussian_mixing}
\end{equation}
whose mean $\mu_\tau$ and standard deviation $\sigma_\tau$ are some smooth time-dependent functions.

\begin{theorem}
\label{sec:A_Theorem}
The quantum map $\Lambda_{\tau}$ obtained by averaging over $\xi$ after a time $\tau$, i.e.:
\begin{equation}
    \rho_\tau = \Lambda_\tau[\rho_{0}] = \mean{p(\xi)}{\Unitary_{\xi \tau}[\rho_{0}]} = \int \!d\xi\, p_{\mu_\tau,\sigma_\tau}(\xi) \;e^{-i\xi \hat{H} \tau} \rho_{0} \ e^{i\xi \hat{H}\tau},
    \label{eq:effective_map}
\end{equation}
corresponds to the solution of the master equation:
\begin{align}
    \frac{d\rho_{\tau}}{d\tau}
    & = -i\omega(\tau) [\hat{H},\rho_\tau] + \Gamma(\tau) \left(\hat{H} \rho_\tau \hat{H} - \frac{1}{2} \{\hat{H}^2,\rho_\tau \} \right) \\
    & = -i\omega(\tau) [\hat{H},\rho_\tau] - \frac{1}{2} \Gamma(\tau) \left[\hat{H},[\hat{H},\rho_\tau]\right],
    \label{eq:VonNeumann1}
\end{align}
where the effective time-dependent frequency and decay parameters read, respectively:
\begin{align}
 	\omega(\tau) = \mu_\tau + \tau \, \Dot{\mu}_\tau
 	\qquad\text{and}\qquad
     \Gamma(\tau) = 2\sigma_\tau^2 \, \tau \left( 1 + \frac{\Dot{\sigma_\tau}}{\sigma_\tau} \, \tau \right).
     \label{eq:effective_freq_decay_pars}
\end{align}
\end{theorem}

\begin{proof}
By explicitly differentiating $\rho_{\tau}$ defined in \eqref{eq:effective_map} with respect to $\tau$, we obtain:
\begin{align}
    \frac{d\rho_{\tau}}{d\tau}
    & =
    \frac{d}{d\tau} \Lambda_{\tau}[\rho_{0}] = \frac{d}{d\tau} \mean{p(\xi)}{\Unitary_{\xi\tau}[\rho_{0}]} = \frac{d}{d\tau} \left( \int d\xi \, p_{\mu_{\tau},\sigma_{
    \tau}}(\xi) \ \Unitary_{\xi\tau}[\rho_{0}]  \right) \nonumber \\
     & = \int d\xi \left( \frac{d}{d\tau} p_{\mu_{\tau},\sigma_{\tau}}(\xi) \right) \Unitary_{\xi\tau}[\rho_{0}] + \int d\xi \ p_{\mu_{\tau},\sigma_{\tau}}(\xi) \frac{d}{d\tau} \Unitary_{\xi\tau}[\rho_{0}] \nonumber \\
     & = \frac{1}{\sigma^3_{\tau}} \;\mean{p(\xi)}{\left( (\xi - \mu_{\tau})\sigma_{\tau} \Dot{\mu}_{\tau} + ((\xi - \mu_{\tau})^2 - \sigma_{\tau}^2)\Dot{\sigma}_{\tau} \right)\Unitary_{\xi\tau}[\rho_{0}]}
     - i[\hat{H},\mean{p(\xi)}{\xi \Unitary_{\xi\tau}[\rho_{0}]}] \nonumber \\
     & = \frac{\Dot{\mu}_{\tau}}{\sigma_{\tau}^2} \;\mean{p(\xi)}{\xi \Unitary_{\xi\tau}[\rho_{0}]} + \frac{\Dot{\sigma}_{\tau}}{\sigma_{\tau}^3} \;\mean{p(\xi)}{(\xi - \mu_{\tau})^2 \Unitary_{\xi\tau}[\rho_{0}]} - \left( \frac{\mu_{\tau} \Dot{\mu}_{\tau}}{\sigma_{\tau}^2} + \frac{\Dot{\sigma_{\tau}}}{\sigma_{\tau}} \right) \rho_{\tau} \nonumber\\
     &\quad - i[\hat{H},\mean{p(\xi)}{\xi \Unitary_{\xi\tau}[\rho_{0}]}],
     \label{eq:master}
\end{align}
where by using relations for the moments of a Gaussian distribution we can further simplify the following expressions:
\begin{align}
    \mean{p(\xi)}{\xi\,\Unitary_{\xi\tau}\!\left[\rho_{0}\right]}
    & = -\mathrm{i}\sigma_{\tau}^{2}\tau\,[\hat{H},\rho_{\tau}]+\mu_{\tau}\rho_{\tau}, \\
    \mean{p(\xi)}{(\xi-\mu_{\tau})^{2}\Unitary_{\xi\tau}\!\left[\rho_{0}\right]}
    & = -\sigma_{\tau}^{4}\tau^{2}\left[\hat{H},[\hat{H},\rho_{\tau}]\right]+\sigma_{\tau}^{2}\rho_{\tau}.
\end{align}
As a result, we can write the dynamics \eqref{eq:master} as
\begin{equation}
    \frac{d\rho_{\tau}}{d\tau} = -i (\mu_{\tau} + \Dot{\mu}_{\tau} \tau)\,[\hat{H},\rho_{\tau}] - \sigma_{\tau}^2 \tau \left(1 + \frac{\Dot{\sigma}_{\tau}}{\sigma_{\tau}}\tau \right) \left[\hat{H},[\hat{H},\rho_{\tau}]\right],
\end{equation}
which is the desired form stated above in equation \eqref{eq:VonNeumann1}.
\end{proof}

\begin{corollary}
\label{sec:A_Corollary}
For the case of $N$ spin-1/2 particles evolving for a time $\tau$ according to the dynamics \eqref{eq:VonNeumann1} with $\hat{H} = \hat{J}_y$, $\omega(\tau) = \gmr B$ and
$\Gamma(\tau)=\gamma_y$, i.e.:
\begin{align}
    \frac{d\rho_{\tau}}{d\tau} = -i \gmr B [\hat{J}_y,\rho_{\tau}] - \frac{1}{2} \gamma_y [\hat{J}_y,[\hat{J}_y,\rho_{\tau}]]
\end{align}
with $B$ being constant over the time $\tau$, the effective quantum map describing the evolution $\Lambda_\tau$ in \eqref{eq:effective_map} can be interpreted as mixture of unitary channels with the  Gaussian mixing probability \eqref{eq:Gaussian_mixing} of mean $\mu_{\tau} = \gmr B$ and standard deviation $\sigma_{\tau} = \sqrt{\frac{\gamma_y}{\tau}}$.
\end{corollary}
The above statement can be straightforwardly verified by explicitly computing the effective time-dependent frequency and decay parameters in \eqref{eq:effective_freq_decay_pars} for the choice of $\mu_{\tau} = \gmr B$ and $\sigma_{\tau} = \sqrt{\frac{\gamma_y}{\tau}}$, which consistently simplify to:
\begin{align}
    & \omega(\tau) = \gmr B+\tau\,0=\gmr B, \\
    & \Gamma(\tau) = 2 \frac{\gamma_y}{\tau}\tau \left(1 + \frac{1}{2}\left( \frac{\gamma_y}{\tau}\right)^{-1/2} \left(- \frac{\gamma_y}{\tau^2} \right)
    \left(\frac{\gamma_y}{\tau} \right)^{-1/2}\tau \right) = 2 \gamma_y \left(1 - \frac{1}{2} \right) = \gamma_y.
\end{align}

\subsection{Fisher information of $\mathbb{P}_{B_t}(\vec{\omega}_{<t})$}
\label{sec:bayesian_info_M}
As discussed in the main text, upper-bounding the measurement-record contribution to the BI, i.e.~$J_M$ defined in equation \eqref{eq:J_M}, corresponds effectively---see inequality \eqref{eq:finalJM}---to computing the FI:
\begin{equation}\label{eq:ap_fisherdef}
  \F\!\left[\mathbb{P}_{B_k}(\vec{\omega}_{k}) \right] = \int\! d\vec{\omega}_{k} \, \mathbb{P}_{B_k}(\vec{\omega}_{k}) \left[ - \partial^2_{B_k} \log\left(\mathbb{P}_{B_k}(\vec{\omega}_{k})\right)\right]
\end{equation}
for the probability distribution defined in equation \eqref{eq:bigP_def}:
\begin{align}\label{eq:ap_expandlogbigP}
  \mathbb{P}_{B_k}(\vec{\omega}_{k}) = \frac{1}{p(B_k)}\int\! d\vec{B}_{k-1} \, p(\vec{B}_{k}) \, q(\vec{\omega}_{k}|\vec{B}_{k}).
\end{align}
Note that $\mathbb{P}_{B_k}(\vec{\omega}_{k})$ contains a Gaussian marginal distribution $p(B_k)$ specified in \eqref{eq:ap_totalpBk_OUP}, as well as products of Gaussian distributions:
\begin{equation}
  p(\vec{B}_{k}) = \prod_{j=1}^{k} p(B_j|B_{j-1}) p(B_0)
  \qquad\text{and}\qquad
  q(\vec{\omega}_{k}|\vec{B}_{k}) = \prod_{j=0}^k q(\omega_j|B_j), \label{eq:ap_def_pBk_q(wkBk)}
\end{equation}
where $p(B_j|B_{j-1})$ is the Gaussian transition probability of the OU process defined in \eqref{eq:OUP_dist} with variance $V_P$ given by \eqref{eq:ap_OUPvar}, $p(B_0)$ is the prior Gaussian distribution with zero mean and variance $\sigma_0^2$, while each $q(\omega_j|B_j)$ is the mixing probability introduced within the CS method in equation \eqref{eq:in_omega_out_B}, which is also a Gaussian with mean $\gmr B_j$ and variance
\begin{equation}\label{eq:ap_defVQ}
  V_Q = \frac{\gamma_y}{\delta t}.
\end{equation}
It follows from the definitions in \eqref{eq:ap_def_pBk_q(wkBk)} that the integral in \eqref{eq:ap_expandlogbigP} can be rewritten as a set of nested integrals, i.e.:
\begin{align}\label{eq:ap_expansion of_logbigP}
  & \int\! d\vec{B}_{k-1} \, p(\vec{B}_{k}) \, q(\vec{\omega}_{k}|\vec{B}_{k}) = \int\! d\vec{B}_{k-1} \prod_{j=1}^{k} p(B_j|B_{j-1}) q(\omega_j|B_j) p(B_0) q(\omega_0|B_0) \nonumber \\
  & \quad= q(\omega_k|B_k) \int\! d B_{k-1} \, p(B_k|B_{k-1}) q(\omega_{k-1}|B_{k-1}) \ldots \int\! d B_{1} \, p(B_{2}|B_{1}) q(\omega_{1}|B_{1}) \nonumber \\
  & \qquad\qquad\int\! d B_0 \, p(B_1|B_0) q(\omega_0|B_0)  p(B_0),
\end{align}
which we would like to simplify. To do so, we first need to prove the following lemma:
\begin{lem} \label{sec:lemma}
Let us consider a recurrence relation between $\mathcal{P}_j(B_j)$ and $\mathcal{P}_{j-1}(B_{j-1})$ for $j=0,1,2,\dots$ as a generalised convolution of Gaussian distributions:
\begin{align}\label{eq:ap_3Gaussians}
\mathcal{P}_j(B_j) & = \int\! dB_{j-1} \, \frac{1}{\sqrt{2 \pi V_P}} e^{-\frac{(B_j - B_{j-1})^2}{2 V_P}} \, \frac{1}{\sqrt{2 \pi V_Q}} e^{-\frac{(\omega_{j-1} - \gmr B_{j-1})^2}{2 V_Q}} \, \mathcal{P}_{j-1}(B_{j-1}),
\end{align}
where $V_P,V_Q\ge0$ and $\mathcal{P}_0(B_0) = C_0 \, e^{-\frac{(B_0 - \mu_0)^2}{2 V_0}}$ with some fixed $C_0,V_0\ge0$ and $\mu_0\in\mathbb{R}$.

Then, for all $j\ge1$:
\begin{equation}
	\mathcal{P}_j(B_j) = C_j \, e^{-\frac{(B_j - \mu_j)^2}{2 V_j}}
	\label{eq:sol_rec_rel}
\end{equation}
where the parameters $C_j$, $\mu_j$ and $V_j$ are given as the solution to the following (coupled) recurrence relations:
\begin{align}
	\label{eq:ap_recursiveform_C} C_j & = C_{j-1} \left(2 \pi \left(\gmr^2 V_P + V_Q + \frac{V_P V_Q}{V_{j-1}}\right)\right)^{-1/2} e^{-\frac{(\omega_{j-1}-\gmr \mu_{j-1})^2}{2(V_Q + \gmr^2 V_{j-1})}} \\
  	\label{eq:ap_recursiveform_mu} \mu_j & = \frac{V_Q \mu_{j-1} + V_{j-1} \gmr \, \omega_{j-1}}{V_Q + \gmr^2 V_{j-1}} \\
  	\label{eq:ap_recursiveform_V} V_j & = V_P + \frac{V_Q V_{j-1}}{V_Q + \gmr^2 V_{j-1}}.
\end{align}
\end{lem}
\begin{proof}
As for any recurrence problem, it is sufficient to prove that the solution \eqref{eq:sol_rec_rel} holds for $j=0$, and that the recurrence relation \eqref{eq:ap_3Gaussians} is fulfilled for any $j\ge1$. The first part is trivially satisfied by definition, while to prove the latter we substitute $P_{j-1}(B_{j-1})$, defined according to \eqref{eq:sol_rec_rel}, into \eqref{eq:ap_3Gaussians} and explicitly perform the integration, i.e:
\begin{align}
    \mathcal{P}_j(B_j) & = \int\!dB_{j-1}\, \frac{1}{\sqrt{2 \pi V_P}} e^{-\frac{(B_j - B_{j-1})^2}{2 V_P}}  \frac{1}{\sqrt{2 \pi V_Q}} e^{-\frac{(\omega_{j-1} - \gmr B_{j-1})^2}{2 V_Q}} C_{j-1} e^{-\frac{(B_{j-1}-\mu_{j-1})^2}{2 V_{j-1}}} \nonumber \\
     \label{eq:ap_jthcase} & = \frac{C_{j-1}}{\sqrt{2 \pi \left(V_Q + \gmr^2 V_P + \frac{V_P V_Q}{V_{j-1}}\right)}} \exp{\left\{ -\frac{B_j^2 -2\alpha B_j + \beta}{2\left(V_P + \frac{V_Q V_{j-1}}{V_Q + \gmr^2 V_{j-1}} \right)} \right\}},
\end{align}
where $\alpha$ and $\beta$ are constants independent of $B_j$ and $B_{j-1}$, and equal to
\begin{align}
	\alpha & = \frac{V_Q \, \mu_{j-1} + V_{j-1} \, \gmr  \, \omega_{j-1}}{V_Q + \gmr^2 V_{j-1}}, \label{eq:alpha}\\
    \beta & = \frac{V_Q \, \mu_{j-1}^2 + V_{j-1} \, \omega^2_{j-1}+V_P(\omega_{j-1}-\gmr \, \mu_{j-1})^2}{V_Q + \gmr^2 V_{j-1}}. \label{eq:beta}
\end{align}
Hence, by `completing the square' we rewrite \eqref{eq:ap_jthcase} as
\begin{align}\label{eq:ap_pj}
    \mathcal{P}_j(B_j) & = \frac{C_{j-1} e^{ -\frac{-\alpha^2 + \beta}{2\left(V_P + \frac{V_Q V_{j-1}}{V_Q + \gmr^2 V_{j-1}} \right)} }}{\sqrt{2 \pi \left(V_Q + \gmr^2 V_P + \frac{V_P V_Q}{V_{j-1}}\right)}}  \exp{\left\{ -\frac{(B_j -\alpha)^2}{2\left(V_P + \frac{V_Q V_{j-1}}{V_Q + \gmr^2 V_{j-1}} \right)} \right\}},
\end{align}
and, after substituting for $\alpha$ and $\beta$ from (\ref{eq:alpha}-\ref{eq:beta}), we arrive at the expression \eqref{eq:sol_rec_rel} for $P_{j}(B_{j})$ with $C_j$, $\mu_j$ and $V_j$ specified by the recurrence relations (\ref{eq:ap_recursiveform_C}-\ref{eq:ap_recursiveform_V}).
\end{proof}

Now, using the above lemma we may rewrite equation \eqref{eq:ap_expansion of_logbigP} as
\begin{align}
   & \int d \vec{B}_{k-1} \, p(\vec{B}_{k}) \, q(\vec{\omega}_{k}|\vec{B}_{k}) = q(\omega_k|B_k) \mathcal{P}_k(B_k),
\end{align}
with $\mathcal{P}_k(B_k)$ being now defined according to equation \eqref{eq:sol_rec_rel} with variances $V_P$ and $V_Q$ in (\ref{eq:ap_recursiveform_C}-\ref{eq:ap_recursiveform_V}) specified in our case by equations \eqref{eq:ap_OUPvar} and \eqref{eq:ap_defVQ}, respectively. Consequently, the FI of $\mathbb{P}_{B_k}(\vec{\omega}_{k})$ in equation \eqref{eq:ap_fisherdef} reads
\begin{align}
  \F\!\left[\mathbb{P}_{B_k}(\vec{\omega}_{k}) \right]
  & =
  \int\! d\vec{\omega}_{k} \, \mathbb{P}_{B_k}(\vec{\omega}_{k}) \left[ - \partial^2_{B_k}\log\left( \frac{1}{p(B_k)}\int\!d \vec{B}_{k-1} \, p(\vec{B}_{k}) \, q(\vec{\omega}_{k}|\vec{B}_{k}) \right) \right] \nonumber  \\
  & =\int\! d\vec{\omega}_{k} \, \mathbb{P}_{B_k}(\vec{\omega}_{k}) \left[ - \partial^2_{B_k}\log\left(  \frac{q(\omega_k|B_k)}{p(B_k)} \mathcal{P}_k(B_k) \right) \right] \label{eq:log_gaussians} \\
  & = \frac{\gmr^2}{V_Q} - \frac{1}{V_P^{(k)}} + \frac{1}{V_k}, \label{eq:ap_fishersimp}
\end{align}
where the expression \eqref{eq:ap_fishersimp} follows from the fact that we are dealing with a product (and quotient) of Gaussian distributions within $\log(\dots)$ in \eqref{eq:log_gaussians}, so that the FI becomes just the sum (and difference) of the inverses of their respective variances. In particular, $V_Q/\gmr^2$ is the variance of $q(\omega_k|B_k)$ when treating $B_k$ as the random variable, $V_P^{(k)}$ is the variance of $p(B_k)$ specified in \eqref{eq:ap_totVarP}; while $V_k$ is the variance of $\mathcal{P}_k(B_k)$ given by the recurrence relation \eqref{eq:ap_recursiveform_V} that must still be solved.

Although the recursive relation \eqref{eq:ap_recursiveform_V} (with $V_0 = \sigma_0^2$) admits a general solution, for simplicity we present only its form for the relevant setting, in which we do not possess any prior knowledge about the initial field $B_0$, i.e.~when $\sigma_0 \to \infty$ in \eqref{eq:B0_dist}. Then,
\begin{align}\label{eq:ap_solrec}
  \left.V_k\right|_{\sigma_0\to\infty} = \frac{1}{2 V_P}  + \frac{1}{2\gmr} \sqrt{V_P (4 V_Q + V_P \gmr^2)} \left(1 + \frac{2}{-1 + \left( \frac{2 V_Q + \gmr\left( V_P \gmr + \sqrt{V_P (4 V_Q + V_P \gmr^2)} \right)}{2 V_Q + \gmr \left(V_P \gmr - \sqrt{V_P (4 V_Q + V_P \gmr^2)} \right)}\right)^k} \right),
\end{align}
which---after substituting for $V_P=\frac{q_B}{2\chi} (1 - e^{-2 \chi \delta t})$ and $V_Q=\gamma_y/\delta t$ according to expressions \eqref{eq:ap_OUPvar} and \eqref{eq:ap_defVQ}, respectively, and taking the continuous-time limit $\delta t \to 0$ with $k=t/\delta t$---takes the form:
\begin{equation}\label{eq:limit_Vk}
	\lim_{\delta t\to0}\left\{\left.V_{k=t/\delta t}\right|_{\sigma_0\to\infty}\right\}
	 = \sqrt{\frac{\gamma_y \, q_B}{\gmr^2}}\, \coth\!\left( t \sqrt{\frac{q_B \gmr^2}{\gamma_y}}\right).
\end{equation}

Finally, by noting that the term $\gmr^2/V_Q = \gmr^2 \delta t/\gamma_y$ in \eqref{eq:ap_fishersimp} vanishes when letting $\delta t\to0$, and so does the term $1/V_P^{(k)}$ when $\sigma_0 \to \infty$ (see equation \eqref{eq:ap_totVarP}), we arrive at the FI of $\mathbb{P}_{B_t}(\vec{\omega}_{<t})$ defined within the continuous-time $\delta t \rightarrow 0$ limit as
\begin{equation}\label{eq:ap_ultimateFisher}
  \F\!\left[\mathbb{P}_{B_t}(\vec{\omega}_{<t}) \right]
  \;\underset{\sigma_0 \to \infty}{=}\;
  \frac{1}{\lim_{\delta t\to0}\left\{\left.V_{k=t/\delta t}\right|_{\sigma_0\to\infty}\right\}}
  = \sqrt{\frac{\gmr^2}{\gamma_y q_B}} \, \tanh\!\left( t \sqrt{\frac{q_B \gmr^2}{\gamma_y}}\right),
\end{equation}
which is the expression stated in equation \eqref{eq:CS_bound_tanh} of the main text.

\section{Effective dynamics averaged over the field fluctuations}
\label{sec:field_avg}
We use the results obtained above in \ref{sec:CS_decomposition}, in order to answer the question of what would the effective ensemble dynamics be, if one was \emph{not} to use inference techniques such as the Kalman filter in order to track the magnetic field in real time, but rather ignore and, hence, average over the field fluctuations.

For that purpose, let us ignore the impact of continuous measurement on the atomic ensemble, and focus on the ensemble dynamics dictated only by the unitary evolution:
\begin{equation}
\label{eq:unitary_ev}
  \frac{d \rho_t}{dt} = - i \gmr B_t \,[\hat{J}_y,\rho_t],
\end{equation}
where $B_t$ follows a Wiener process $dB_t = \sqrt{q_B} \, dW_t$, such that $\E[dW_t^2] = dt$. Given that the ensemble is initially prepared in a pure state $|\psi_0\rangle$, its state at time $t$ then reads
\begin{equation}
\label{eq:state@t}
  |\psi_t \rangle = e^{- i \gmr \hat{J}_y \int_{0}^{t} B_\tau \, d\tau  } |\psi_0\rangle.
\end{equation}
The integral of a Wiener process, $\int_{0}^{t}W(\tau)\,d\tau$ with $W(t):=\int_{0}^{t}\,dW_t$, constitutes a random variable distributed normally with zero mean and variance $t^3/3$~\cite{Gardiner1985}. Hence, upon defining
\begin{equation}
\label{eq:redefBt}
  Z_t : = \frac{\gmr}{t} \int_{0}^{t} B_\tau \, d\tau
\end{equation}
that satisfies then $Z_t\sim\mathcal{N}(0,\gmr^2 q_B t/3)$, we may rewrite \eref{eq:state@t} as
\begin{equation}\label{eq:rho_redefined}
   |\psi_t\rangle  = e^{- i \, Z_t \, t \, \hat{J}_y} \, |\psi_0\rangle,
\end{equation}
and the quantum state describing the atomic ensemble after averaging over the field fluctuations reads
\begin{align}\label{eq:average_overZt}
  \rho_t = \E_{Z_t}[|\psi_t\rangle \langle \psi_t|] = \E_{Z_t}\left[e^{- i \, Z_t \, t \, \hat{J}_y} \, |\psi_0\rangle \langle \psi_0| \, e^{i \, Z_t \, t \, \hat{J}_y}\right],
\end{align}
being averaged over all potential stochastic trajectories of the variable $Z_t$.

Now, inspecting Theorem \ref{sec:A_Theorem} and noticing that \Eref{eq:average_overZt} is just a special case of \Eref{eq:effective_map}, we may directly conclude from \eref{eq:VonNeumann1} that the average dynamics is described by the following master equation
\begin{equation}
\label{eq:mastereq_avgsigma}
  \frac{d\rho_t}{dt} = - \frac{1}{2} \Gamma(t) [\hat{J}_y,[\hat{J}_y, \rho_t]] = \Gamma(t) \mathcal{D}[\hat{J}_y]\rho_t,
\end{equation}
with the resulting decoherence rate given by
\begin{equation}
\label{eq:gamma}
  \Gamma(t) =\gmr^2 q_B \, t^2,
\end{equation}
which is obtained by substituting into \Eref{eq:effective_freq_decay_pars} $\sigma_t^2 = \gmr^2 q_B t / 3$, i.e.~the variance of the random variable $Z_t$.

The above short derivation demonstrates that averaging over field fluctuations---instead of following a single trajectory, as done by resorting to, e.g., the Kalman filter---effectively leads to a collective noise in the direction of magnetic field (here, in the eigenbasis of $\hat{J}_y$), whose decoherence rate \eref{eq:gamma}, however, increases quadratically with time, $\Gamma(t)\sim t^2$. This contrasts the collective noise model considered throughout the manuscript, whose decoherence rate is constant, $\gamma_y=1/T_2^*$ in \Eref{eq:cond_dyn}, being determined by the phenomenological (ensemble) spin-decoherence time $T_2^*$.



\newcommand{\newblock}{}
\bibliographystyle{myapsrev4-1}

\bibliography{noisy_mag}

\end{document}